\documentclass{article}
\usepackage{arxivnopreprint}
\usepackage{amsmath,amsthm,amssymb}
\usepackage{algorithm}
\usepackage{algpseudocode}
\usepackage{subcaption}
\usepackage{array}
\usepackage{graphicx,import,color}
\usepackage{adjustbox}
\usepackage{rotating}
\usepackage{csvsimple}
\usepackage{geometry}
\usepackage[normalem]{ulem}
\RequirePackage{doi}
\usepackage{hyperref}
\newcommand{\reals}{\mathbb{R}}
\newcommand{\cons}{\mathtt{Constraints}}
\newcommand{\prior}{\mathtt{Priorities}}
\newcommand{\deadline}{\mathtt{Deadlines}}
\newcommand{\earliest}{\mathtt{Earliest}}
\newcommand{\sched}{\mathtt{Schedule}}
\newcommand{\naturals}{\mathbb{N}}
\newcommand{\pushback}{\mathtt{PushBack}}

\newtheorem{definition}{Definition}[section]
\newtheorem{heuristic}{Heuristic}[section]
\newtheorem{lemma}{Lemma}[section]
\newtheorem{theorem}{Theorem}[section]
\newtheorem{assumption}{Assumption}[section]
\newtheorem{example}{Example}[section]
\newcount\Comments  
\newcommand{\kibitz}[2]{\ifnum\Comments=0\textcolor{#1}{#2}\fi}

\begin{document}

\title{Real Time Operation of High-Capacity Electric Vehicle Ridesharing Fleets}
\author{
    Matthew Zalesak \\
    School of Operations Research and Information Engineering \\
    Cornell University \\
    Ithaca, NY 14853 \\
    \texttt{mdz32@cornell.edu}\\
    \And
    Samitha Samaranayake \\
    School of Civil and Environmental Engineering \\
    Cornell University \\
    Ithaca, NY 14853 \\
    \texttt{samitha@cornell.edu}
}

\maketitle

\begin{abstract}
We study the feasibility of using electric vehicles in online, high-capacity ridepooling systems.  Prior work has shown that online algorithms perform well for centrally-controlled, high-capacity ridepool systems.  First, we propose a mixed integer linear program to expand past algorithms on ridepooling to electric vehicles fleets with the objective of scheduling vehicle charging to maintain sufficient fleet sizes at various times of day.  Then we show a faster, scalable algorithm with similar performance that is practical for full-scale systems.  Our contributions show the importance of having knowledge and estimates of future demand even when operating in the online setting.
\end{abstract}

\section{Introduction}

Recent innovations in communications technology and mobile devices have enabled the emergence of large-scale ridehailing\footnote{We use the term \emph{ridesharing} to refer to any system that allows riders to travel in vehicles not owned by themselves.  For an explanation of terminology used in this paper, refer to section \ref{sec:terminology}.} services which provide on-demand, door-to-door transportation services.  These services are attractive to users as they provide alternatives to the cost and hassle of finding parking, overcome the the problems of slow or unreliable transit options, and offer service when no other transit options are available \cite{clewlow2017disruptive}.  These new services are possible as private operators begin to take advantage of new mobile technologies and methodologies to match riders and drivers together with higher speed and efficiency than has been done before.  Cities also stand to benefit from ridehailing services with the hopes that they may act as an extension to public transportation, reduce congestion, and have a positive impact on pollution.   

While private ridehailing services have made a debut in many countries around the world, in many cities these services have actually been found to contribute to higher congestion \cite{erhardt2019transportation} and green house gas emissions \cite{rodier2018effects}.  One cause for this is asymmetry in demand patterns causing local mismatches between supply and demand, leading many vehicles in ridehailing services to deadhead (also known as rebalancing) in order to relocate to other areas \cite{han2017quantification}.  This process increases the number of vehicle miles traveled required to serve the traveling population compared to the same users driving in their own vehicles.  A solution to this problem that has received increased attention recently is high-capacity ridepooling, where a large number passengers can share a trip in a single vehicle in order to increase the system's efficiency (ratio of passenger miles traveled to vehicle miles traveled).  An ever growing area of research, shared ridepooling has the potential to offer mass mobility as a compliment to existing public transit options.  As a consequence of serving riders using fewer vehicle miles traveled, shared ridepooling has the potential to reduce roadway congestion and reduce pressure in municipal parking areas \cite{viegas2016shared}.

In addition to using shared ridepooling as a tool to reduce pollution, total greenhouse gas emissions can be further reduced by operating ridepool systems with fleets of hybrid or electric vehicles rather than conventional internal combustion engine (ICE) vehicles, with most of the benefits of electric vehicles in particular depending on the availability of green energy sources~\cite{onat2015conventional}.  Not only do electric vehicles help reduce smog and other pollutants in cities, empirically they have been found to be more energy efficient in urban settings due to the ability of regenerative braking to reclaim energy.  With ongoing research on autonomous vehicles making large autonomous electric fleets a possibility, it is necessary that algorithms for ridepooling systems adapt to handle the new constraints introduced by electric fleets.

Past research on ridesharing systems has generally focused on either large-scale, high-capacity ridepooling or smaller scale and lower-capacity ridehail systems with electric vehicles, but few have focused on both.  Our contribution with this work is a model and accompanying methods for the control and operation of large-scale, high-capacity ridepool systems with a fleet comprised of electric vehicles\footnote{Our approach can handle the case of mixed fleets by setting the battery capacity of ICE vehicles to be infinite or explicitly modeling refueling if necessary.}.  Unlike models that have assumed the use of ICE vehicles, electric vehicle fleets must be able to go out of service to charge mid-day to maintain quality levels of service.  Even as maximum battery ranges continue to grow in higher end cars, the actual number of miles that electric vehicles can travel can be significantly affected by outside factors such as temperature \cite{lajunen2018evaluation}.  In addition, time dependent consumption of power from air and auxiliary systems mixed with slow travel in congested city roadways can negatively impact range~\cite{restrepo2014performance}.

The paper is organized as follows.  Section \ref{sec:litreview} reviews the literature on ridehailing and ridepooling algorithms as well as studies concerning electric vehicles.  Section \ref{sec:prelim} formally defines the setting of the problem and reviews an algorithm for ridepool management that we extend for electric vehicles.  Section \ref{sec:methods} presents the formulation of the charge scheduling subproblem and proposes two solution methods.  Finally, section \ref{sec:results} presents the results of simulations of large-scale ridepool systems and compares the performance of the proposed solutions with benchmark and ICE vehicle baselines.

\subsection{Literature Review}
\label{sec:litreview}

Literature on operational algorithms for high-capacity ridepooling and management of electric vehicles in ridepooling systems has largely been separate.  Much of the work on the management of large-scale on-demand ridepooling assumes the use of conventional internal combustion engine (ICE) vehicles.

Many early ideas for large-scale ridepooling employed greedy assignment algorithms.  This group of methodologies consider requests sequentially, usually in chronological order, and assigns them to vehicles in a way that minimizes some criteria subject to constraints.  \cite{ma2013t} presented T-Share, a framework that manages a real-time taxi system where multiple users can share a ride.  The system sequentially handles incoming travel requests by matching each request with the vehicle that would have to travel the least additional distance to serve them.  The focus is on how to efficiently compute this greedy assignment.  In follow up work, \cite{ma2014real} demonstrated simulations of the system with over 7000 taxis on a map in Beijing.  To measure the performance of the system, they introduce the taxi-share rate and seat occupancy rate statistics, which other works use as well, to benchmark the performance gains of ride-sharing algorithms.  The taxi-share rate is the percentage of requests that share a ride while the seat occupancy rate is the ratio of passenger-minutes to the total number of available seat-minutes in the system.

\cite{ota2016stars} proposed the STaRS simulation framework which, similar to T-share, incorporates quality of service requirements limiting the maximum amount of waiting time for customers that are served as well as the total amount of delay that can be incurred due to sharing rides.  Requests arrive into the system online and are processed sequentially, each request is greedily assigned to the best vehicle, and vehicles are not rebalanced at the end of each iteration.

Another line of methodologies instead rely on batching requests over a time period and processing requests simultaneously.  Many of these methods aim for exhaustive searches over the space of vehicle and request matches during each batch.  
Starting this line of work, \cite{santi2014quantifying} introduced the notion of a shareability network.  The shareability network is a graph that contains a node for every request in the system.  If two requests are mutually shareable under quality of service requirement such as maximum waiting time and detour time, an edge is drawn between them.  For groups of 3 or more requests, shareability can be represented by checking group feasibility and then creating hyperedges connecting the request nodes.

\cite{alonso2017demand} built upon the concept of shareability network in a batch-optimization framework.  In this framework, requests arrive to the system in batches.  A shareability network is produced on the set of requests that arrived in the batch.  Unlike the work of \cite{santi2014quantifying}, it also includes nodes representing the vehicles in the shareability network.  By modifying the definition of a trip to be a collection of passengers along with a vehicle, every feasible shareable trip forms a clique in the shareability network including exactly one vehicle.  By fixing a particular vehicle and group of requests it is possible to compute optimal routing, costs, and revenues associated with various assignment and routing decisions.  Using the results of the extended shareability network, the set of trips and costs are used as inputs to integer linear program.  The result of decomposing the problem into trip generation and selection preserves optimality in the solution up to computational limits that apply at each step.  The results is a system that makes routing decisions that are optimal at all times given request currently in the system.

Other papers propose ideas that straddle the space between greedy assignment and batching.  A variant proposed by \cite{simonetto2019real} allows for significant speed gains at the cost of optimality.  To avoid searching for cliques in a sharability graph, they restrict the definition of a trip to be only a vehicle and a single request.  At each time step, the optimization only requires solving a linear assignment problem.  Once a passenger is assigned to a vehicle, they cannot be reassigned to a different vehicle later even if the new matching would lead to better performance.  In addition, \cite{lowalekar2019zac} consider decomposing the problem into paths rather than into trips which allows for faster computation at the cost of exact constraint evaluation.

Work on electric vehicle systems address two main challenges: where to locate charging infrastructure and how to operate fleets of electric vehicles in real time.

On the front of locating charging infrastructure, most researchers propose two phase procedures where a first stage simulates a system initialized without charging infrastructure, and then use these results to make decisions about locating the charging infrastructure via various methods. For example, the works of \cite{chen2016operations}, \cite{farhan2018impact}, and \cite{loeb2018shared} simulate a ridehail system with electric vehicles with a first phase that dynamically creates charge stations at the location of each vehicle if the vehicle is running too low on charge and no previously created charge station is in range. Using a different approach, \cite{zhang2020charging} instead saves the locations of the charge requests when vehicles report they are low on charge and then solves a K-means problem to locate the charging stations.  Going beyond the use of simulations to provide the two phases, \cite{kohani2017location} studies where to place charging stations by reviewing GPS data for vehicle trajectories from real world systems and using where they tend to park and spend time idling to guide station location decisions.

Some researchers have studied charge station location outside of the context of a two-phase simulation.  For example,  \cite{lam2014electric} looks to choose the best location for charge stations out of a set of candidates by solving a MIP model that minimizes the cost of constructing charge stations subject to coverage constraints.  Similarly, \cite{lee2019shared} chooses charge station locations from a set of candidates by solving a p-median problem minimizing the average distance from nodes in the map to charge stations.  

On the electric vehicle operations side of the problem, many studies have focused on simple myopic policies \cite{fehn2019modeling, bauer2019electrifying, li2019agent, bongiovanni2018two} while others have attempted to incorporate planning for future demand \cite{al2020approximate, shi2019operating, kullman2020dynamic, iacobucci2019optimization}, though these methods do not necessarily scale to operational size.

Many prior works have used myopic threshold charging policies or policies activated by special events.  These policies allow the vehicles to operate without regard to future charging needs until vehicle charge goes below a threshold value or the range is low enough that it cannot reach a charge station after serving a possible next request.  The method of \cite{fehn2019modeling} assigns passengers to the nearest vehicle and charges the vehicles when they get below 25\%, or a different threshold at some times of day based on the cost of electricity.  In a similar vein, \cite{bauer2019electrifying} also uses a greedy passenger assignment algorithm and a greedy algorithm to assign vehicles to charge stations when they need to charge.  Using a slightly different mechanism, \cite{li2019agent} implements charging as part of a rebalancing process where the vehicle charges when it cannot reach its rebalancing target with its current level of charge.  Planning charging as part of a vehicle's travel assignment, \cite{bongiovanni2018two} sequentially processes requests that arrive to the system and attempts to insert them into the vehicle's schedule, with strategies to add, extend, or modify the vehicle's charge station visits if necessary.

Some studies have used methods that try to be aware of the future, either directly or indirectly.  
One approach is to use approximate dynamic programming (ADP), such as \cite{al2020approximate} which uses ADP to determine when vehicles get new passengers and whether vehicles should charge.  In \cite{shi2019operating} and \cite{kullman2020dynamic}, deep reinforcement learning is used to develop policies for vehicles to determine when to accept new customers and when to charge.  They suggest that the learning process allows the system to anticipate future demand.  In \cite{iacobucci2019optimization}, two model predictive control optimization algorithms are used to both solve for routing and charging decisions, demonstrating their method on up to 30 vehicles.

A few studies have used other methods that were applied at small scale, such as \cite{shi2018routing} which controls 4 single capacity vehicles and determines assignments and charging times using a MIP model and \cite{pettit2019increasing} which looks for a policy to operate a single electric vehicle using reinforcement learning.

In this work we will derive methods from the batching, exhaustive search framework of \cite{alonso2017demand} to show performance under optimal routing with high-capacity electric vehicles is possible in real time without having to use a myopic charging strategy.  Nothing we present prevents adapting the underlying fleet assignment and routing algorithm to other methods in the literature as our electric vehicle additions are modular in this respect.

\section{Preliminaries}
\label{sec:prelim}

\subsection{Terminology}
\label{sec:terminology}

The area of shared transportation does not, unfortunately, have universally standardized terminology and so we briefly explain our usage in this paper.  We use the term \emph{ridesharing} to refer to any system that allows riders to travel in vehicles not owned by themselves.  This includes carpooling systems where the driver of the vehicle is traveling to the same or similar destination as the passengers as well as settings where the driver is strictly offering services to passengers.  When we wish to draw attention to the fact that the driver of a vehicle is working as a driver rather than making personal travel, we call it a \emph{ridehailing} system.  When we wish to draw attention to the fact that multiple passengers can travel in a vehicle simultaneously, we call it a \emph{ridepool} system.

\subsection{The Problem}

In the general on-demand ridepool problem, the inputs given are a graph $G = (N, A)$ representing a road network, a set of vehicles $V$ with possible heterogeneous capacities that are centrally controlled by the operator, and a set of passengers $R$ that arrive online.  Each passenger request $r \in R$ is specified by a system entry time, an origin, and a destination.

The system is constrained to provide a set of minimum quality of service (QoS) requirements to all requests that are accepted into the system.  These requirements are general, can be request specific, and can be chosen differently by each operator.  For example, an operator may impose a maximum waiting time and a maximum detour constraint, ensuring that all accepted passengers are picked up within $t_w^{\max}$ time of submitting their requests and that the extra time traveling in the vehicle due to a shared ride is no more than $t_s^{\max}$.  An operator may opt to choose settings such as $t_w^{\max} = 5$ minutes and $t_s^{\max} = 15$ minutes.

If a request cannot be served using a trip that meets the QoS requirements, it must be rejected.  The feasibility of a trip is determined by solving a constrained traveling salesman problem, where among other constraints we must consider the capacity of the vehicle at each point in time.  Thus, matching a passenger to a vehicle in a particular batch does not require that the passenger will actually board the vehicle during that batch or that the vehicle has capacity for the new passenger at that precise moment in time.  In deterministic settings\footnote{deterministic travel times and no dropouts}, since we gain no new information over time that reduces the set of constraints in the system (i.e., no information that allows for previously infeasible pickups to be feasible) if a request is rejected in one batch, it will be infeasible to find a trip to serve them in later batches as well.

In our problem we assume that vehicles are electric and identical with battery properties that will be discussed in section \ref{subsec:battery}.  A subset of the nodes in the graph are designated as charge stations, $S \subset N$.  Each charge station $s \in S$ has a maximum capacity $K_s$ denoting the maximum number of vehicles that can simultaneously charge at $s$.  The charge stations are private and are for the exclusive use of the operator; no consideration is given to others using the stations.

Operationally, the system is different from typical ridepool systems since 1) electric vehicles need to go out of service during the middle of the day to charge and 2) charging electric vehicles takes significantly longer than filling an internal combustion engine vehicle with gasoline.\footnote{Operationally, refueling does not have a large impact on the performance of ICE vehicles since they usually have a long fuel range, can refuel in a short time, and have a large number of options for getting fuel.} It is possible that with poor planning many of the charge stations could become congested and not be able to serve all the vehicles that are assigned to charge at them.  If it is not possible for a vehicle to continue operation due to low charge, it may go to a charge station and wait.  When vehicles charge, we assume they will always charge up to level $q_{\max}$, which may or may not be the vehicle battery's theoretical maximum capacity.

Though the system operates online, without a prediction of future demand we cannot make meaningful decisions about when to charge vehicles, especially since some may need to be preemptively charged.  For example, during rush hour (when demand is at its highest) it would be unwise to maximally utilize the charging infrastructure.  To this extent, we allow the vehicle operator to create, in advance, a requirement function $R(t)$ for each time $t$ designating the number of vehicles that must be online rather than charging.  The operator can compute this function in any way, be it from demand quantity alone, value of customers at certain times of day, etc.  It can also be created in more complicated ways such as considering the cost of electricity by time of day and making the requirement function so as to limit vehicles being offline and charging during the most expensive times of day.

The computation of $R(t)$ may hide a great deal of work itself.  However, not only are these values computed offline, these values can involve a large number of practical considerations from the operator that are beyond the scope of this paper.  In this work we limit the construction of $R(t)$ to simple functions of expected demand, as discussed in section \ref{subsec:availability}.

The objective of the operator is to assign requests to vehicles and assign routes to vehicles to maximize some function of the assignments.  If the function is linear in the selection of request and vehicle assignment decisions, this leads to an ILP formulation for the problem.

\subsection{Trip-Oriented ridepooling Framework}

The underlying passenger assignment framework used in this work is based on the trip-oriented formulation of \cite{alonso2017demand}.  The online arrival of requests is split into batches of uniform length and computations are performed once for every batch.  For any vehicle $v \in V$ and any set of requests $r_b \subseteq R_b$ all in the same batch, a trip $t = (v, r_b)$ is a pairing of the vehicle and the requests.  For any given trip, an optimal route for the vehicle and associated cost $c_t$ can be found by solving a constrained traveling salesman problem (CTSP).  The method used in \cite{alonso2017demand} for trip generation is very general and allows for the insertion of additional constraints, something we will take advantage of as we extend the problem to electric vehicles.

The formulation of \cite{alonso2017demand} is presented as an optimization formulation that first assigns a large penalty for requests that are not assigned and second penalizes the cost of serving the assigned trips.  They use a binary variable $\epsilon_r$ to track which requests are rejected and then apply a large penalty, $M$.  For the secondary objective, the cost $c_t$ is left unchanged.  The optimization problem remains general, however, since one could choose $M = 0$ and then use an arbitrary function to assign cost $c_t$, possibly negative, for each trip.  The operator's assignment problem in a non-EV ridepool system would thus be the solution of
\begin{align*}
    \min_x & \sum_{t \in Trips} c_t x_t + M \sum_{r} \epsilon_r \\
    \text{subject to} 
    & \sum_{\substack{t \in Trips : \\ v \in t}} x_t \leq 1 & \forall v \in V \\
    & \sum_{\substack{t \in Trips : \\ r \in t}} x_t + \epsilon_r = 1 & \forall r \in R_b \\
    & x \in \{0, 1\}^{|Trips|}, \epsilon \in \{0, 1\}^{|R_b|}
\end{align*}
where here $Trips$ is the set of feasible trips and the solution is constrained to assign at most one trip to each vehicle and choose at most one trip containing each request. Notice that this relies on full enumeration of all feasible trips.  The authors of \cite{alonso2017demand} observe that when the quality of service requirements are tight, as is typical in ridepooling applications, the set of feasible trips is small enough to approximately enumerate in real-time.

Trips are found using a shareability graph, a concept introduced in 
\cite{santi2014quantifying}.  A shareability graph, as extended by \cite{alonso2017demand}, is an undirected graph $G = (N, E)$ in which the node set $N = V \cup R$ is comprised of nodes for each vehicle and request.  The edge set $E$ contains all pairs $(v, r)$, $v \in V, r \in R$ if vehicle $v$ can serve $r$ and all its current passengers while satisfying the quality of service constraints.  $E$ also contains all pairs $r_1, r_2 \in R$ if it is possible for an ideal hypothetical vehicle to serve both requests (for example, if the requests are 5 minute apart and the maximum waiting time is 2 minutes even an ideally located vehicle cannot serve both).

Every feasible trip induces a clique in the shareability graph containing exactly one vehicle node.  This is the case since the edge set was constructed from necessary conditions for the feasibility of a trip.  While the general problem of finding cliques is hard, \cite{alonso2017demand} proposes a method for searching for the cliques in the shareability graph that is exact and, when heuristics are used such as restricting requests to be paired with the nearest 30 vehicles, can be solved in real time.

It may be the case that after the assignment process some requests are not served and some vehicles are left unassigned.  If this is the case, a minimum cost matching based on distance is performed between the vehicles and the rejected requests.  Those vehicles are then routed to the locations of the rejected requests as a heuristic for rebalancing to areas that require additional coverage.  This is only one strategy for rebalancing; more sophisticated rebalancing strategies have also been proposed \cite{liu2019proactive, wallar2019optimizing}.

\subsection{Extending Models for Electric Vehicles} \label{sec:extending}

Algorithms for managing ridepooling fleets with electric vehicles are differentiated from typical ridepooling algorithms in that we must now additionally plan for when and where the vehicles will charge according to our operational objective. While many ridehailing operators use business models where individuals choose when to enter and exit the market, something that would naturally provide charge scheduling outside of the scope of the operator, our work considers algorithms that scale for use with high-capacity vehicles, possibly up to capacity 10.  Higher capacity vehicles increase the potential benefits of pooling, but come with the caveat that these vehicles are unlikely to be owned by individuals.  Our choice to use a centrally controlled system where drivers always follow the directions of the operator, as is the case in the model we build off of from \cite{alonso2017demand}, reflects this expectation.  Therefore, we aim to create a charging schedule for the vehicles given fixed charging infrastructure with the assumption that vehicles will respect their assignments.

Rather than computing charging schedules and passenger assignments in a single computation, we separate the operations of the system into two steps.  The first component, from the trip-oriented framework, assigns requests to vehicles constrained on not violating the current charging schedule.  The second component plans when and where vehicles will charge constrained on the previously created request assignments.  As shown in Figure (\ref{fig:ev_overview}) and algorithm (\ref{alg:basic}), they work in successive steps, each making decisions fully consistent with the previous output of the other.

\begin{figure}
    \centering
    \includegraphics[scale = 0.5]{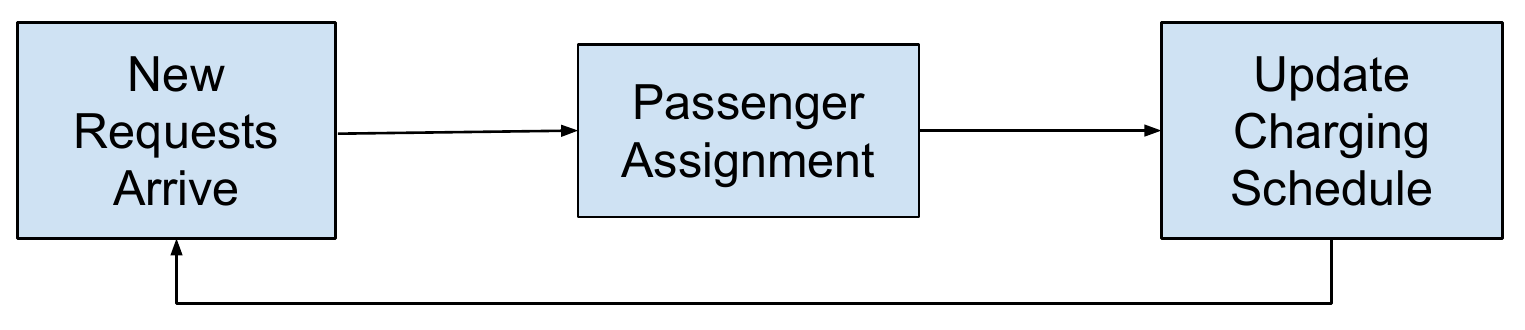}
    \caption{Flow or each iteration.  First new requests arrive.  Constrained on the existing charge schedule, passenger assignments are made.  Finally, conditioned on the new passenger assignments the charging schedule is updated.}
    \label{fig:ev_overview}
\end{figure}

\begin{algorithm}
  \caption {Basic Electric Vehicle Ridepooling}
  \label{alg:basic}
  \begin{algorithmic}[1]
    \State Initially run charging planner.
    \For {each iteration (occurs once every $T^B$ time)}
        \State Collect list of new requests and vehicle states.
        \State Run the request assignment algorithm.
        \State Run the charge planning algorithm.
    \EndFor
  \end{algorithmic}
\end{algorithm}

Implementing the constraints on passenger assignments is straightforward in the trip-oriented assignment framework.  All that must be done is add extra constraints in the trip generation phase to exclude possible tours that violate the charging constraints while solving the constrained-TSP.  Because of this, the rest of the paper is dedicated to formulating the last step - where the charge schedule is generated and updated.

\subsection{Battery Model}
\label{subsec:battery}

An important part of running any system with electric vehicles is to understand how the batteries behave.  In this section we briefly review some of the literature on batteries used in electric vehicles and then demonstrate that simplified models can be justified for typical everyday use in the context of operational optimization of ridepooling fleets.

\subsubsection{Discharging.}

Discharging involves drawing energy from the battery to power the vehicle's drive train and auxiliary systems.  Detailed research has been done on the manner in which electric vehicles consume power based on topography, acceleration, regenerative breaking, etc.  For high accuracy modeling of power consumption, one may refer to works such as \cite{fiori2016power}.

We model discharge, the decrease in battery state of charge, with a linear function of the distance traveled.  This is a reasonable first order approximation since the drive train of the vehicle consumes more power than the auxiliary systems.  If one wishes to go further, models of power consumption could be used to define a custom arc ``distance" using as a measure the typical power consumption of vehicles over those arcs, as done in works such as \cite{lin2016electric}.  More generally, since we always know the path the vehicle has traveled on or is planned to travel on, any general path-based model can be used in our work, though we continue to use the simple distance model for the purposes of our numerical results.

It should be noted that while modern electric vehicles have increasingly long range, practical considerations such as heating, air conditioning, radio, and windshield operation can significantly reduce the range of the vehicle.  While the engine accounts for a majority of energy consumption, according to \cite{restrepo2014performance} using heating can shorten the range of an electric vehicle by 46.4 km per hour of operation.  This is not insignificant in many large ridehailing markets such as New York City.  Similar but smaller effects for air conditioning exist as well.

In order to schedule vehicles to charge one must be able to estimate when the vehicle will run out of power.  The primary variable source of power consumption in a vehicle is the engine, whose energy consumption is a function of distance and terrain.  However, in an online system it is difficult to predict the mileage and terrain vehicles will travel on.  Therefore we use predicted or historical data to estimate the typical rate, $q_{est}$, at which vehicles consume power over time from both time dependent auxiliary services and non-time dependent components such as the engine.  This estimate is based not only on the particular fleet of vehicles but also on the operator's experience with the particular road network they operate on.  Since actual vs predicted power consumption is always stochastic, it is preferable for $q_{est}$ to err on the side of over estimating power consumption over time.  This helps alleviate the risks associated with sudden bursts of power usage not accounted for by the scheduling algorithm.

\subsubsection{Charging.}

While it may be natural to assume that batteries gain charge linearly over time when they are plugged into charge stations, this turns out to not be the case \cite{pelletier2018charge}.  Many batteries charge using a procedure called Constant-Current Constant-Voltage.  Under this system, the voltage in the charger slowly increases to maintain a constant current into the battery.  However, batteries have a maximum safe voltage they can operate at.  Once the voltage in the charger reaches that threshold, the voltage is then held constant at the maximum value and the current into the battery steadily drops off.  This explains why batteries can charge to mid-high charge values, such as 75\%,  quickly though getting to 100\% can take significantly longer.

The differential equation that governs the state of charge in the battery in the constant voltage regime is given by
\[
        \frac{d}{dt}\text{SOC}(t) = \frac{V_{CV} - V_{OC}(\text{SOC}(t))}{R \cdot 3600 \cdot Q}
\]
where SOC is the state of charge, $V_{CV}$ is the voltage in the charger at the constant voltage level, $V_{OC}$ is the open circuit voltage, $R$ is the resistance in the charger, and $Q$ is the maximum battery capacity in ampere-hours.  Simplifying the form of the differential equation, we get
\[
    \frac{d}{dt}\text{SOC}(t) = C_1 - C_2 V_{OC}(\text{SOC}(t)).
\]
If we make the rough approximation that $V_{OC}$ is a linear function of the state of charge, then the solution to the differential equation is
\[
    \text{SOC}(t) = \frac{C_1}{C_2} + k e^{-C_2 t}.
\]

Since under this model the battery can never theoretically be fully charged, we make the following assumption.

\begin{assumption}
\label{assumption:fullcharge}
The manufacturer specifies that ``100\%" charge is some fixed level below this asymptotic maximum.  Meaning, ``100\%" is not truly at full charge.
\end{assumption}  

As an illustration of the model, Let $T$ denote the time to fully charge an empty battery to the manufacturer's full capacity (less than true theoretical capacity), let $R$ denote the number of minutes after which the charging rate of an initially empty battery begins to decrease and let $Q$ be the amount of charge in the battery at time $R$.  Figure (\ref{fig:charging_graph}) shows a typical charging curve.

\begin{figure}[ht!]
  \centering
  \includegraphics[scale=0.6]{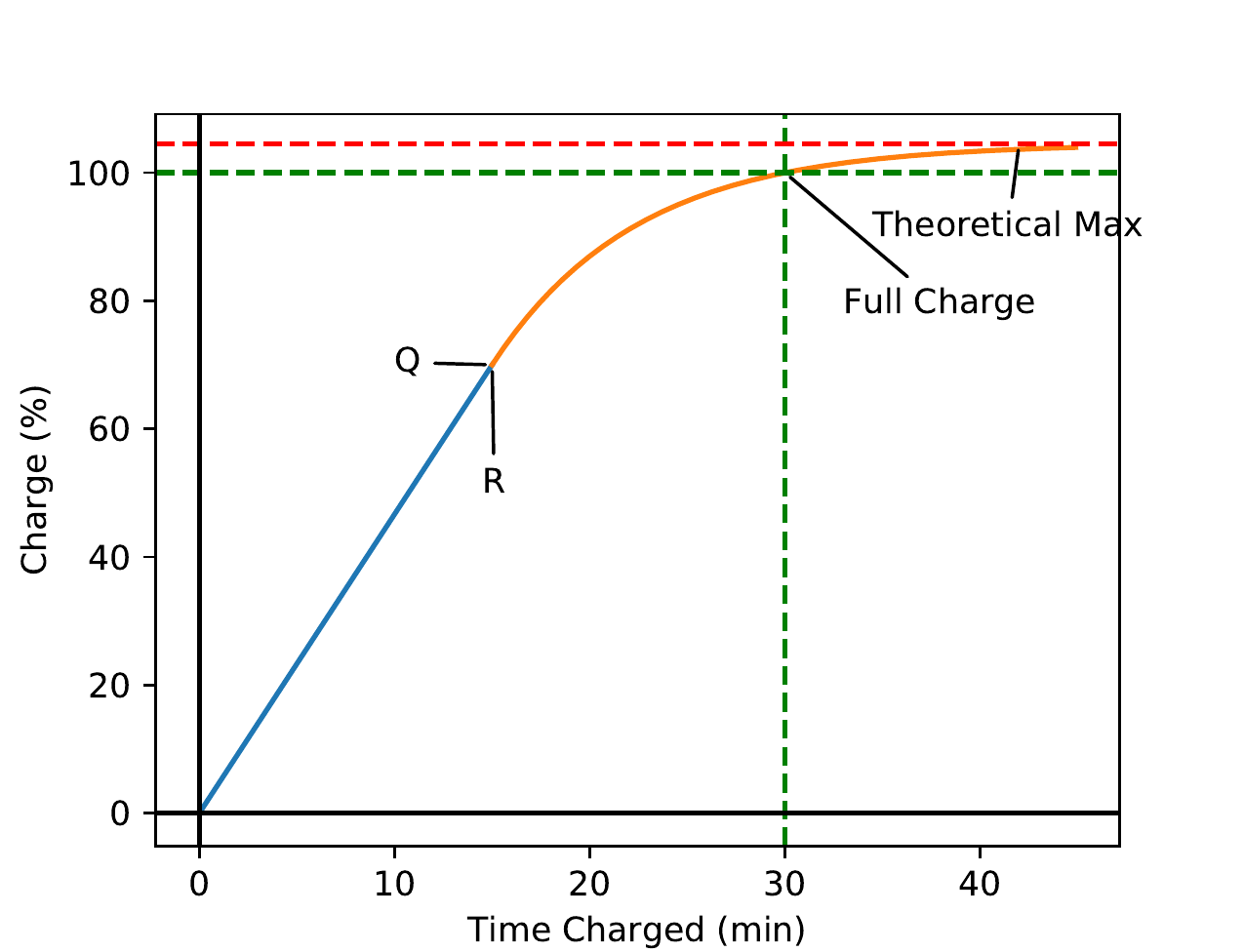}
  \caption{Simple model of charge in an initially empty battery as a function of charging time.  The curve has two components: in blue, the charge in the battery initially increases linearly until time $R$, at which point the battery has charge $Q$ and charging continues as the orange curve which asymptotically approaches the battery's theoretical maximum charge at an exponentially decreasing rate.  Denoted in green are the lines $x = T$ and $y = 1$.  The asymptotic limit is shown as a red line.  In this example, linear charging lasts for the first $R = 15$ minutes at which point the battery has charge $Q = 0.7$ and it takes $T = 30$ minutes for a full charge.}
  \label{fig:charging_graph}
\end{figure}

Due to assumption \ref{assumption:fullcharge}, the charge in the initially empty battery as a function of time is given by the concave function
\[
  \boldsymbol{c_0(}t\boldsymbol{)} = 
  \begin{cases}
    \frac{Qt}{R} & 0 \leq t \leq R \\
    1 - \frac{Q}{R\beta}\left(\frac{1}{e^{\beta(T - R)}}\right)\left(e^{\beta(T - t)} - 1\right) & R < t
  \end{cases}
\]
where the scaling parameter $\beta$ is given by
\[
  \beta = \frac{Q}{R (1 - Q)} + \frac{1}{T - R} W
      (z e^z), \quad z = \frac{R - T}{R}\frac{Q}{1-Q},
\]
where $W$ denotes the Lambert W function and $\frac{Q}{R} > \frac{1}{T}$.

Now consider a battery that is not initially empty.  If the vehicle requires $t$ minutes to completely charge, the amount by which the battery was initially depleted must have been $c(t) = 1 - \boldsymbol{c_0(}T - t\boldsymbol{)}$.  Similarly, given a battery with state of charge $q$ the time to fully charge is given by $f(q) = c^{-1}(q)$.

Since charging is less time efficient as the battery nears full charge, it is natural from an economic standpoint that there may be a policy that vehicles should only charge up to some level $q \in [Q, 1]$ instead of fully charging due to the opportunity costs associated with charging, such as inability to carry passengers during charging.  Additionally, since charging is most efficient in the linear charging region with state of charge in $[0, Q]$ we make the following assumption about when vehicles begin charging.
\begin{assumption}
Vehicles always deplete their battery to a state of charge less than or equal to $Q$ before they begin charging.
\end{assumption}

Consider the following scenario.  Suppose an operator allows the batteries in the fleet's vehicles to discharge down to zero at which point the operator orders the vehicles to charge their batteries back up to state of charge $q'$.

For simplicity, suppose that $q_{est}$ is the discharge rate while vehicles operate and that vehicles have to travel empty for $d$ minutes to reach a charge station.  If the cost per minute of unusable vehicles is $c$, then the average operating cost is

\[
  k(q') = \frac{d + f(q')}{d + f(q') + q'/q_{est}} \cdot c.
\]

The optimal value for $q'$ can be quickly found via binary search.  Figure \ref{figure:marginalcosts} shows examples of choosing a charge up-to policy of minimum cost.  The first shows a general example and the second shows the effect of having no lost time ($d = 0$) traveling with no passengers while they travel to charging stations.

\begin{figure}[h!]
    \centering
    \begin{subfigure}[b]{0.4\textwidth}
      \includegraphics[width=\textwidth]{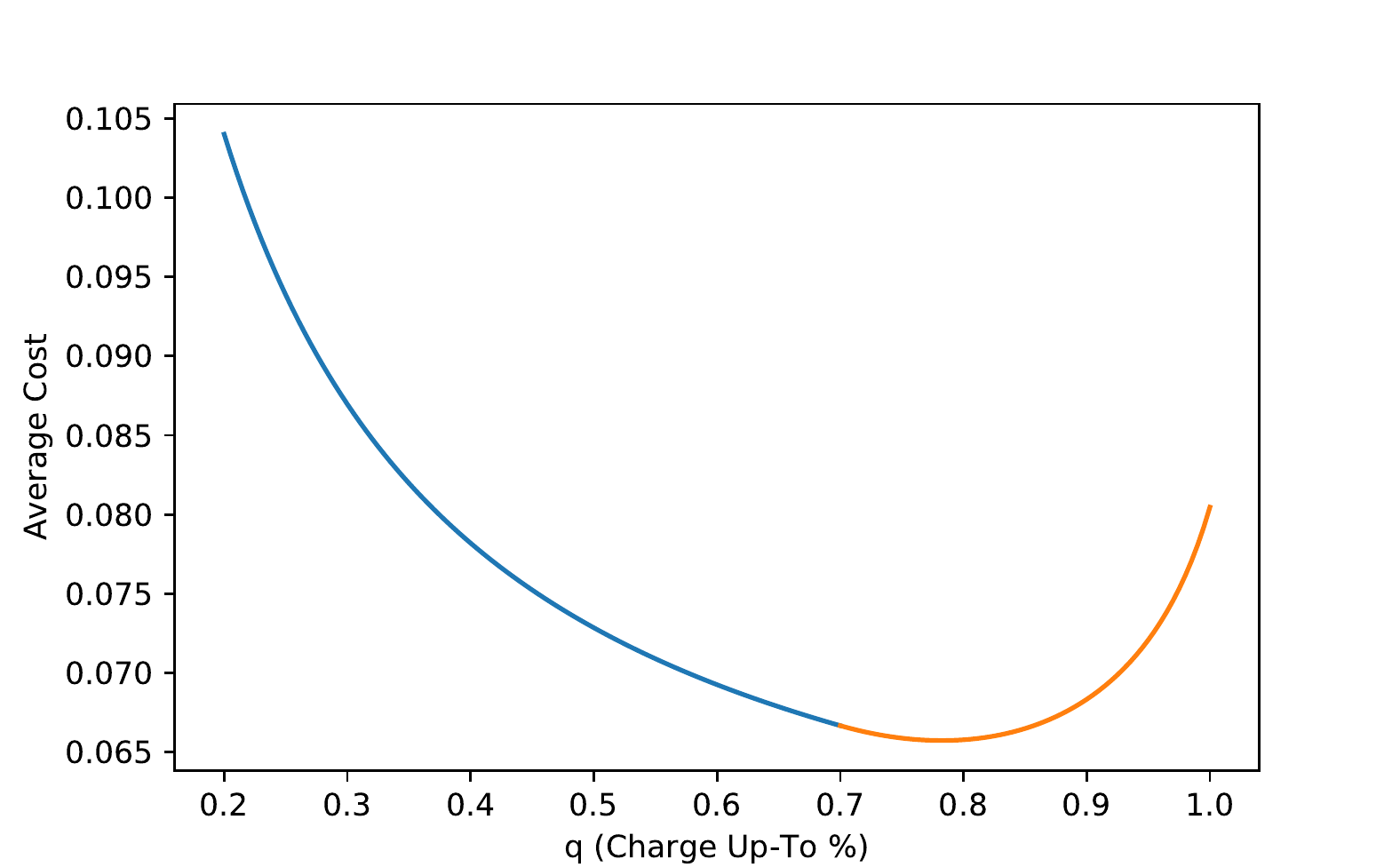}
      \caption{General case.}
      \label{fig:marginalcost_a}
    \end{subfigure}
    \begin{subfigure}[b]{0.4\textwidth}
      \includegraphics[width=\textwidth]{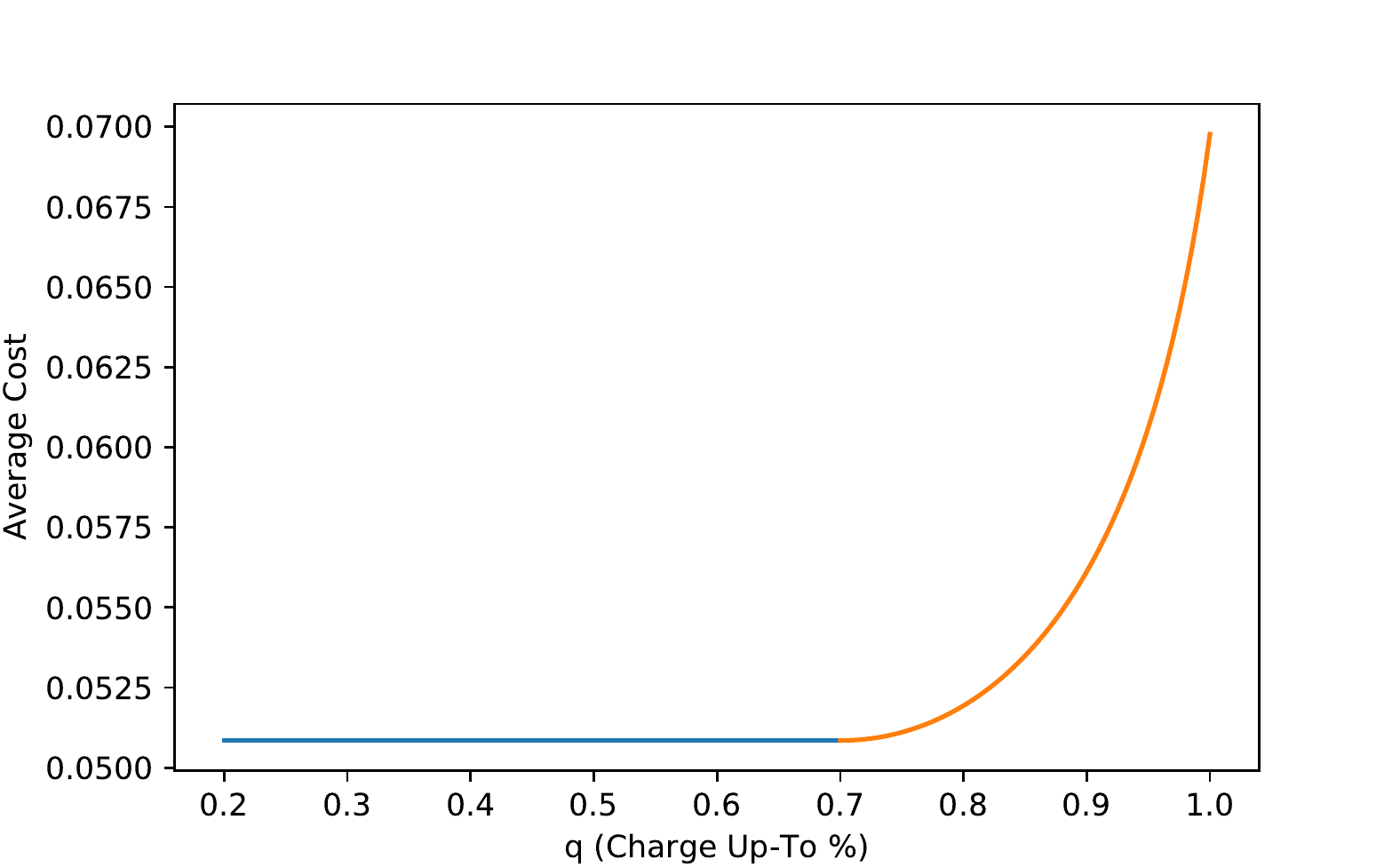}
      \caption{No fixed penalty per charge.}
      \label{fig:marginalcost_b}
    \end{subfigure}
    
    \caption{Assume a vehicle takes 30 minutes to fully charge and receives 70\% of its charge from empty in the first 15 minutes.  Suppose that on average the vehicle wastes 5 minutes of time traveling to the charge station while not holding any passengers, a full battery lasts 400 minutes, and that offline time costs the operator \$1 per minute.  Figure (\ref{fig:marginalcost_a}) plots $k(q')$ with $d = 5$ and Figure (\ref{fig:marginalcost_b}) plots $k(q')$ with $d = 0$.  Optimal value in general case is $q = 78.2\%$ and optimal value when $d = 0$ is any value less than $Q$.}
    \label{figure:marginalcosts}
\end{figure}

In the general case, see Figure (\ref{fig:marginalcost_a}), the optimal behavior is to allow the vehicles to charge slightly beyond the constant-voltage portion of the charging curve.  Referencing the corresponding charging curve in Figure (\ref{fig:charging_graph}), charging just beyond the constant voltage point, $Q$, can be approximately modeled as a linear function with little loss of accuracy.

\subsubsection{Battery Conclusions} \label{sec:batteryconclusions}

While battery charging functions are non-linear, under optimal operating policies they are nearly linear.  This optimal policy is in contrast to that used in studies such as \cite{froger2017new} that assume full recharging and find that the non-linear portion of the charging curve is important to model accurately.

In the rest of this paper we assume that charging times are linear with rate $\eta$ (\% / minute).  The operator fixes a charge-up-to policy under which we relabel charge levels so 100\% denotes the maximum level the operator is willing to charge vehicles to.  We also make the following important assumption about the operator's battery policy.
\begin{assumption}\label{assumption:qmin}
The operator fixes a buffer level $q_{\min}$ which is the lowest charge vehicles are allowed to have.
\end{assumption}
For example, this might be 15\%.  If the vehicle accidentally goes slightly below this level it will still be able to reach a charge station without being towed.  Not only does this help extend the life of the battery, this assumption is important since it will give us a buffer in scheduling vehicles to charge as predicting battery usage under our discharge model cannot be fully predicted since vehicle routing can constantly be changed and updated.  For convenience, we will denote this as 0\% since it is the lowest we intend to let the vehicle discharge to.

\section {Methods}
\label{sec:methods}

In this section we propose algorithms for solving the charge scheduler's problem.  Recall from section \ref{sec:extending} that the scheduler needs to determine when each vehicle will charge and where.  Before we present our proposed algorithms, we first make some notes on the interaction between passenger assignment and charge scheduling.

The charge scheduler and the passenger assigner interact through the necessity of mutually compatible solutions~- passenger assignments cannot prevent a vehicle from charging and an updated charge schedule cannot interfere with a passenger's ride.  In terms of scheduling, it is sufficient to have a mutually agreed method for computing release dates for each vehicle.  The release date for a vehicle is the earliest time the vehicle can charge given its current location and passenger assignments.  When scheduling vehicles for charging, this calculation is used to avoid scheduling charging while the vehicle is serving a passenger.  When considering possible passenger assignments, this determines if a potential trip routing violates the charging schedule.

We define release dates for two cases - when the vehicle's assigned charging station is known and when the assigned charging station is unknown.

Suppose vehicle $v$ is assigned to charge at station $s$.  Let $t_c$ and $n$ denote the time and location $v$ drops off the last passenger under a potential or actual trip.  The earliest time the vehicle $v$ can be scheduled to charge at station $s$ is
\[
    e_{v, s} = t_c + t(n, s),
\]
where $t(a, b)$ denotes the time needed to travel from node $a$ to $b$ in the graph.

In the other case, if the charge station is unknown we have to intelligently add a time buffer.  Let $s^{opt}$ be the closest charge station to the vehicle's ending location $n$.  Given buffer time $D$, we can conservatively estimate the earliest charging time as
\[
    e_v = t_c + \max\{t(n, s^{opt}), D\}.
\]

The choice of $D$ directly impacts the system as it limits which requests vehicles can be paired with whenever the charge station assignment is not yet known.  It would be conservative to choose $D$ to be the worst case travel time between any charge station and node in the graph that are the furthest apart since this will always guarantee each vehicle can be assigned a station to charge at at the assigned time.  In practice one would likely choose a smaller value such as the time sufficient to reach the nearest $k$ charge stations since this would allow more riders to be matched at the small risk that vehicles periodically cannot reach their assigned station.  In such cases, the vehicle could pause and charge after being reassigned stations or, in the worst case, be directed to a closer charge station where it can charge later after some delay.

\subsection{Charge Scheduler Problem Formulation}

The objective of the charge scheduler is to charge vehicles so that, to the extent possible, there are at least $R(t)$ vehicles on the road at any given time.  Let the day be separated into a set of discrete time intervals $T$.  Assume that at each time $t$ we expect we need $R(t)$ vehicles available and the actual number of vehicles available is $B(t)$.  Using the positive part operator $(x)^+ = \max\{0, x\}$, the objective of minimizing total short fall in supply is
\begin{equation}
    \min \sum_{t \in T} \big(R(t) - B(t)\big)^+. \label{eqn:shortfall}
\end{equation}

We allow the loss to be linear in the number of vehicles short we are in the solution in contrast to making a hard constraint.  This is to ensure that a feasible solution always exists.  Additionally, the soft penalty takes a linear form since it is more efficiently computed than other penalties, such as quadratic penalties.

To make the objective well defined, we must specify when a vehicle is and is not available.  Notice that as a vehicle approaches the time it needs to charge, the set of requests it can be assigned to becomes continually more restricted.  This means vehicles on the road may be neither fully available nor unavailable.

Consider the following example:

\begin{example}
Suppose there is a vehicle with a single passenger who will be dropped off in 5 minutes.  When the vehicle drops off the passenger it will be 5 minutes away from the charge station it has been assigned to go to.  There are currently 15 minute remaining until the vehicle must be at the charge station.

In this example, after the vehicle drops off the passenger it will have 5 minutes of idle time where one is naturally temped to say the vehicle is available.  However, if a request arrives that wishes to be driven 10 minutes in a direction away from the charge station, the vehicle will have to reject this request.  On the other hand, if a request arrives asking to travel 4 minutes and wishes to be dropped off close to the charge station the vehicle may be able to serve it.  Therefore the vehicle is neither fully available or fully unavailable.
\end{example}

This justifies the introduction of fractional availability.  Fractional availability is a value between 0 and 1 that is a function of the time remaining until a vehicle is scheduled to charge and of the time remaining since a vehicle completed its previous charging.  When a vehicle is charging the availability is 0.  Likewise, when the previous charging and next charging times are far in the past and future, the availability should be 1.  It is also reasonable that for a vehicle that charges exactly once, the availability should be monotone non-increasing as the charging time approaches, and be monotone non-decreasing afterwards as the charging time has passed.

To incorporate such a general function into an MILP formulation of the scheduling problem, we introduce the function $A(t)$ to denote the availability of a vehicle that charges for one time period starting at time $t = 0$, ending at $t = \Delta$.  Then, for any charging schedule the actual availability of each vehicle can be computed as the minimum over such functions with appropriately adjusted arguments.  As an example, one might choose availability function
\[
    A(t) =
        \begin{cases}
            \min(1, -\delta t)          & t \leq 0 \\
            0                           & 0 \leq t \leq \Delta \\
            \min(1, \delta(t - \Delta)  & t \geq \Delta.
        \end{cases}
\]

Formally, $A(t) : \reals \to [0, 1]$ is monotone non-increasing for $t < 0$, is zero for $0 \leq t \leq \Delta$, and is monotone non-decreasing for $t > \Delta$.  Let the time periods a vehicle is scheduled to charge at be given by the binary variables $x_t$, for times $t \in T$.  The availability of the vehicle at any time $t$ is given by the equation
\[
    a(t) = \min_{s \in T} 1 - (1 - A(t - s)) x_s.
\]

In the context of a MILP, suppose for each vehicle $v \in V$ and each discrete time $t \in T$, the continuous variable $a_{v, t}$ denotes the availability of vehicle $v$ at time $t$ and the binary decision variable $c_{v, t}$ determines if $v$ charges at time $t$.  Since without loss any optimal solution will maximize availability, it is sufficient to use the constraint
\[
    a_{v, t} \leq \min_{s \in T} 1-  (1 - A(t - s)) c_{v, s}.
\]

Thus, since the value of $A(t - s)$ is not dependent on any decision variable even general nonlinear availability functions can be incorporated into the MILP.

\subsection {Two Stage Planning Process}

Our method for charging management is split into two parts: time scheduling and location scheduling.  Here we will briefly justify and explain the details of this separation.

\begin{figure}[b!]
    \centering
    \includegraphics[scale = 0.5]{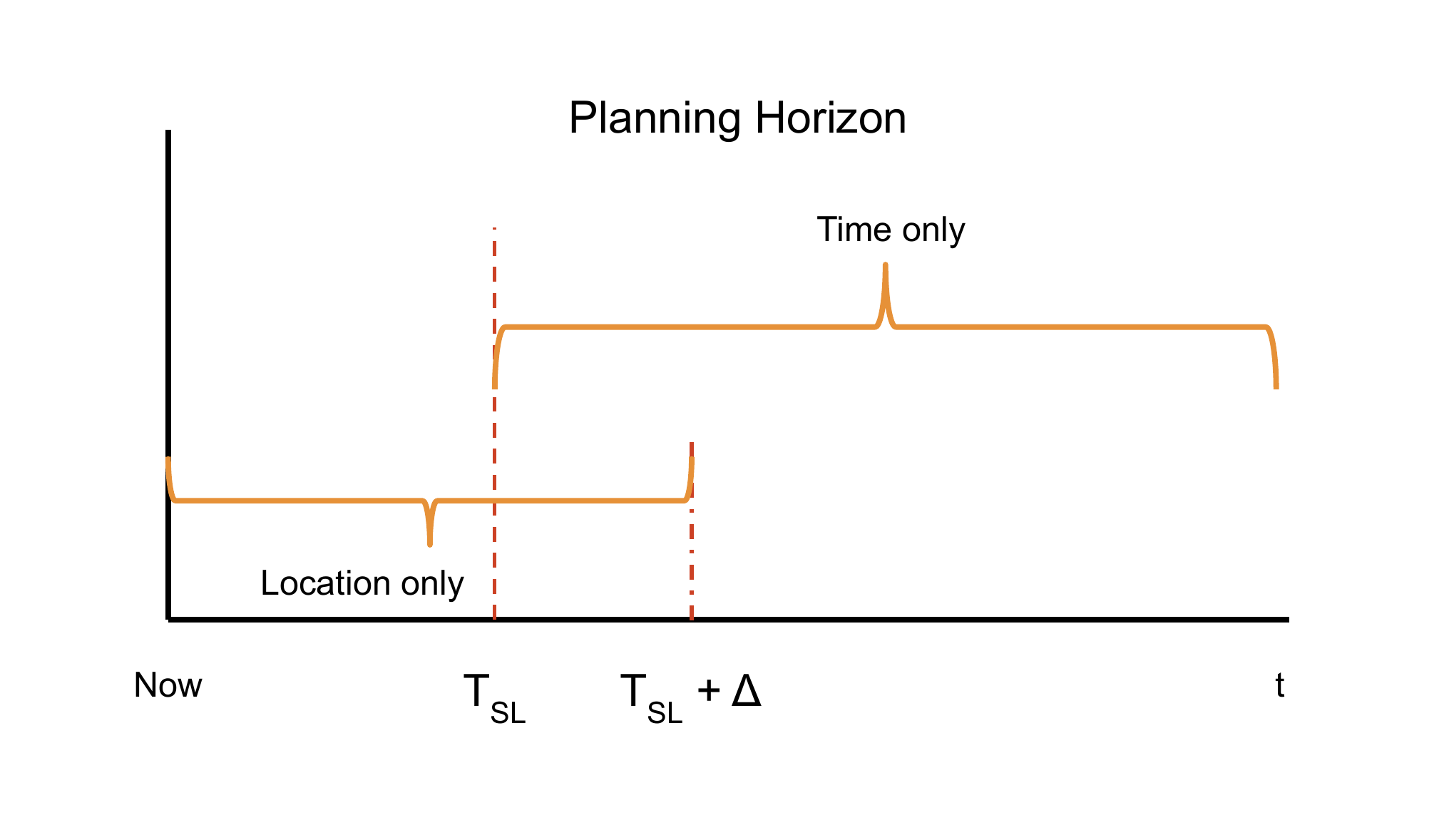}
    \caption{In the long horizon, the schedule more than $T_{SL}$ time in the future can be modified.  The short horizon station problem only assigns vehicles that are scheduled to charge within $T_{SL} + \Delta$ of current time.  $\Delta$ can be 0, but theorem \ref{thm:deltalimit} shows it cannot be too large.}
    \label{fig:longshort}
\end{figure}

Scheduling vehicles to charge has two components: determining when vehicles will charge and where they will charge.  Since ridepooling systems operate online, choosing the charging location for vehicles many hours in advance cannot be meaningfully done in any `optimal' way, other than ensuring each vehicle will have a station to charge at.  With this observation, we split the charge scheduling process into two separate parts: a long horizon where we only consider changing the time at which vehicles charge and a short horizon where we only consider the locations at which vehicles charge.  More formally, the long horizon only modifies charge schedules at least $T_{SL}$ time in the future while the short horizon assigns charge stations to vehicles scheduled to charge within $T_{SL} + \Delta$ for some chosen $T_{SL}$ and parameter $\Delta$ controlling the maximum overlap.  For example, one might choose $T_{SL} = 45$ minutes and $\Delta = 15$ minutes.  See Figure (\ref{fig:longshort}).  It is not necessary to choose $\Delta > 0$, but an upper limit on $\Delta$ exists as will be shown in theorem \ref{thm:deltalimit}.

While the proposed split reduces computation complexity of the problem, the solutions it produces may be suboptimal.  For example, it may be possible to choose better station locations in the short horizon if charge times can be shifted by a small amount.  However, since the short horizon problem is defined only over a  small subset of times, the flexibility to reschedule vehicle charging times is greatly restricted unless the full complexity of the long horizon problem is reintroduced.

With the short and long horizon problems being solved separately, it is necessary to take steps to ensure the times vehicles are scheduled to charge in the long horizon problem admit feasible charge station assignments.  The choice of $D$ in the release date equation $e_v = t_c + \max\{t(n, s^{opt}), D\}$ provides the trade off between guarantees of feasibility with conservative, large $D$ and flexibility to serve more requests with aggressive, smaller $D$.  Sufficiently large $D$ guarantee that feasible solutions will exist, but come at the cost of being restrictive on which requests a vehicle can serve.

Since an operator may not feel that updating time and location information for vehicle charging is necessary at the same frequency as passenger assignments are made, it is possible to run these planning algorithms at only a subset of the iterations.  For example, if requests arrive in batches every $T^B$ seconds, then one could run the short horizon station location problem every $T^S$ seconds (divisible by $T^B$) and run the long horizon time problem every $T^\ell$ second (divisible by $T^S$).  As a variant of algorithm (\ref{alg:basic}), this modified procedure is shown in algorithm (\ref{alg:ridesharing}).

\begin{algorithm}
  \caption {Electric Vehicle Ridepooling}
  \label{alg:ridesharing}
  \begin{algorithmic}[1]
    \State Initially run Long horizon planning algorithm.
    \For {each iteration (occurs once every $T^B$ time)}
        \If {$T^\ell$ time has passed since last run of short horizon planner}
            \If {$T^S$ time has passed since last run of long horizon planner}
                \State Run the long horizon planning algorithm.
            \EndIf
            \State Run the short horizon planning algorithm.
        \EndIf
        \State Run the request assignment algorithm.
    \EndFor
  \end{algorithmic}
\end{algorithm}

As a final note, the choice of overlap $\Delta$ must be less than or equal to $2T^\ell$, as explained in the following theorem.

\begin{theorem}
\label{thm:deltalimit}
It is necessary and sufficient that $\Delta \leq 2T^\ell$ to guarantee that feasible charge station assignments will exist for all vehicles.
\end{theorem}

\begin{proof}
Without loss of generality, assume the system clock is currently set to time 0.

\begin{enumerate}
    \item  Case 1 : $\Delta = 2T^\ell$.
    
    Consider the set of vehicles the long horizon problem schedules to charge in time periods $T_{SL}$ and $T_{SL} + \Delta$.  Those in the former will have their assignments fixed and will not appear in the next round of long horizon optimization since the long horizon only modifies schedules at least $T_{SL}$ time in the future.  The later will appear one more time in the next iteration of long horizon optimization.
    
    Any vehicle in the next round of longest optimization will either be a vehicle previously assigned in time slot $T_{SL} + \Delta$, in which case it already has a station assignment, or a vehicle that has no previous station assignment.  In the case of no assignment, the buffer time $e_v = t_c + \max\{t(n, s^{opt}), D\}$ guarantees these vehicles have passenger assignments that are compatible with being assigned any station.  Thus, a feasible charge station assignment always exists.
    
    \item  Case 2 : $\Delta > 2T^\ell$
    
    Let $\Delta = 3T^\ell$.  Suppose there are two charge stations of capacity one that are far apart.  In the first iteration, one vehicle is assigned to charge for a single time period at time $T_{SL} + \Delta$ and another is assigned to charge at time $T_{SL} + 2\Delta$.  Assume both vehicles are assigned the same charge station, call it $c_1$.  Since both of these vehicles now have an assigned station, their buffer time calculation changes to $e_v = t_c + t(n, c_1)$.  This means that it is possible that each vehicle, in the next iteration, is assigned passengers that want to travel to $c_1$ and who will be dropped off just as the vehicles' assigned charge times begins.
    
    In the next iteration, it is feasible for the long horizon problem to now assign both vehicles to the same charge time, $T_{SL} + 2\Delta$.  However, if the distance between $c_1$ and the other station, $c_2$, is greater than $\Delta$, then both vehicles can still only reach $c_1$ due to the passenger assignments.  But they are scheduled to charge at the same time and each station has capacity 1.  So there is no longer a feasible charge station assignment.
\end{enumerate}
\end{proof}

\subsection{Long Horizon Planning}

The long horizon problem sets vehicle charging times that are scheduled at least $T_{SL}$ beyond the present time.  We propose solving this problem using a MILP that assigns charging times for the full day.  In addition to the objective of minimizing short falls given in equation (\ref{eqn:shortfall}), we also include a soft penalty if the schedule expects any vehicles to reach or go below zero charge.  The soft penalty on negative charge levels guarantees there is always a feasible solution to the long horizon problem for any input.\footnote{In practice, the model penalizes vehicles for having negative charge but not for being close to negative charge, so many vehicles will run out of charge slightly before the scheduled charge time due to the discharge function depending on the vehicle's routing, which is not fully known in advance.  This is not a problem since either vehicles can idle at the charge station prior to charging or the `zero' charge mark is actually a safe buffer away from true empty.}  Additionally, recall from earlier that we assume the operator chooses a charge level $q_{\min} > 0$ that we refer to as 0\% charge.  Therefore, even if a vehicle goes slightly below this level it will still be able to operate.

Since we cannot know the exact charge in the vehicles in the future, we estimate the charge by computing the charge in the vehicle when it is done serving its current passengers at time $e_v$, which we call $Q_{v, 0}$.  We then use the estimated linear discharge rate $q_{est}$ discussed in the section on batteries for the time periods the vehicle does not charge and $\eta$ to account for when the vehicle does charge.

The total capacity of all charge stations is denoted by $K = \sum_{s \in S} K_s$.  Since the fractional availability function may be less than one for vehicles that have charged, or are planning to charge, prior to the first time period considered in the long horizon problem, we introduce the value $P_{v, t}$ which equals $A(t - t')$ where $t'$ is the most recent time period vehicle $v$ charged.

There are three decision variables used in the problem.

\begin{enumerate}
    \item $c \in \{0, 1\}^{|V||T|}$ - binary decision to charge vehicle $v$ at time $t$.
    \item $a \in \reals_+^{|V||T|}$ - decision of fractional availability of vehicle $v$ at time $t$.
    \item $q \in \reals^{|V||T|}$ - decision of estimated charge in vehicle $v$ at time $t$.
\end{enumerate}

Using the above, the MILP model\footnote{Let $(x)^+ = \max\{0, x\}$, $(x)^- = \max\{0, -x\}$.} is defined by 

\begin{align}
    \min & \sum_{t \in T} \left(R(t) - \sum_{v \in V} a_{vt} \right)^+ + M \sum_{v \in V} \sum_{t \in T}(q_{vt})^- \label{ilp:long:objective} \\
    \text{s.t.} 
    & \sum_{v \in V} c_{vt} \leq K & \forall t \in T \label{ilp:long:stationcap} \\
    & 
    c_{vt} = 0 & \forall t < e_v \ \forall v \in V \label{ilp:long:tooearly} \\
    & q_{v,t_v} = Q_{v0} & \forall v \in V \label{ilp:long:initialcharge} \\
    & q_{vt} \leq q_{v,t-1} - q_{est}(1 - c_{v,t-1}) + \eta c_{v,t-1} & \forall t > e_v \ \forall v \in V \label{ilp:long:chargeupdate} \\
    & a_{vt} \leq \min_{s \in T} \left(1 - (1 - f(t - s)) c_{vs}\right) & \forall t \in T \ \forall v \in V \label{ilp:long:available} \\
    & a_{vt} \leq A_{vt} & \forall t \in T \ \forall v \in V \label{ilp:long:initialavailability} \\
    & 0 \leq a \leq 1, \quad \quad c \in \{0, 1\}, \quad \quad q \leq 1, \label{ilp:long:variables}
\end{align}
where $e_v$ is as defined at the beginning of section \ref{sec:methods}.

The objective (\ref{ilp:long:objective}) penalizes having fewer than the required number of vehicles at each point in time and also penalizes for vehicles having negative charge.  $M$ allows control between supply shortfall and avoiding negative charge as the primary objective.  $M$ is typically large since it is most important that charge is non-negative whenever possible.

Constraint (\ref{ilp:long:stationcap}) forces the number of vehicles charging at any one time to be at most the total available capacity across all charging stations.
Constraint (\ref{ilp:long:tooearly}) ensures that no vehicle is scheduled to charge before the release date, which accounts for final passenger drop-off and the minimum time needed to reach charge stations.

Constraint (\ref{ilp:long:initialcharge}) defines the charge in the vehicle at the vehicle's release date. 
Constraint (\ref{ilp:long:chargeupdate}) then describes how the charge value evolves over time, decrementing by $q_{est}$ if it is not charging, and increasing by $\eta$ if it is.  The constraint prevents vehicles from gaining more charge in a single period than is possible using the charging infrastructure and works in conjunction with constraint (\ref{ilp:long:variables}) which limits the maximum charge in vehicles to 100\%.

Constraint (\ref{ilp:long:available}) defines the fractional availability of each vehicle in terms of the function $f(t)$, though it does not account for the possibility of reduced availability due to prior charging.
Finally, constraint (\ref{ilp:long:initialavailability}) restricts the value of availability to account for the prior charging. 

While non-preemptive solutions are preferred, the model is preemptive since there are no constraints that require selected charging periods to be of the length necessary to fully charge.  The solution is encouraged to be non-preemptive through the penalty of availability introduced in constraint (\ref{ilp:long:available}).  Since the availability function $a(t)$ reduces the amount of availability a vehicle is considered to have just before and after charging, when a vehicle's charging is preempted the system cannot fully reclaim the vehicle's availability.  Between it and the vehicle that preempted it, this will detract from the systems' total amount of availability and potentially lead to a worse objective value.  In schedules where many vehicles need to charge and the system is running near minimum availability in all time periods, non-preemptive solutions allow the model to charge more vehicles while maintaining the required availability.  Therefore, we expect that typical solutions will not contain many preemptive elements.

\subsection{Short Horizon Planning}

Short horizon planning refers to the problem of scheduling where vehicles should charge when the vehicles' scheduled charge times are within a \textit{short} window, i.e. less than $T_{SL} + \Delta$.  Short horizon planning will make final assignments of vehicles to stations minimizing the impact of charging on the vehicles.  This is done by associating a cost $d_{v, s}$ for matching vehicles $v \in V$ to stations $s \in S$.  For example, the cost $d_{v, s}$ might be the distance from the vehicle's last drop-off point to station $s$.  When it is not feasible to match a particular vehicle and station we let $d_{v, s} = \infty$.

Although the short horizon problem does not select charging times, this problem must nonetheless consider both space and time.  Let $V' \subseteq V$ be the set of vehicles to be assigned charge stations, let $T$ refer to the set of time intervals during which any vehicle in $V'$ is assigned to charge, and recall that $K_s$ is the capacity of each station.  For any time $t \in T$, let $V_t$ be the set of vehicles that charge at time $t$.  For any vehicle $v \in V'$, let $S_v$ be the set of stations the vehicle can reach by its scheduled charge time.

The sole binary decision variables is $P \in \{0, 1\}^{|V'||S|}$ which is the indicator for whether vehicle $v \in V'$ charges at station $s \in S$.  The binary programming formulation is

\begin{align}
\min & \sum_{s \in S, v \in V'} d_{vs} P_{vs} \label{ilp:short:objective} \\
\text{s.t.}
    & \sum_{v \in V_t} P_{vs} \leq K_s & \forall s \in S, \forall t \in T    \label{ilp:short:capacity}  \\
    & \sum_{s \in S_v} P_{vs} = 1 & \forall v \in V'  \label{ilp:short:assign} \\
    & P \in \{0, 1\}^{|V'||S|}
\end{align}

Constraint (\ref{ilp:short:capacity}) ensures that each station is operating at or below capacity during each time period and constraint (\ref{ilp:short:assign}) ensures each vehicle receives exactly one valid station assignment.

The station assignment problem can be shown to be NP-complete by reduction to the hierarchical interval scheduling (HIS) problem with single machine capacities which was introduced and shown to be hard by~\cite{kolen2007interval}

\begin{definition}[from Kolen et al. 2007 ]
Given $K$ continuously available machines and $n$ jobs with starting times $s_j$ and finishing time $f_j$ such that each job can only be processed on machine $\{0, \dots, r_j\}$, it is NP-hard to determine if there exists a feasible schedule.
\end{definition}

\begin{theorem}
    The short horizon station location problem is NP-complete.
\end{theorem}
\begin{proof}
To show the problem is in NP it suffices to show that the number of constraints is polynomial in $|V|$ and $|S|$.  Constraints (\ref{ilp:short:assign}) clearly satisfy this.  It is sufficient for constraint (\ref{ilp:short:capacity}) to satisfy this if the number of discrete time periods $T$ is $O(|V|)$.  Since the starting and ending time of each charging job is fixed, an interval graph can be created with vehicles as vertices and edges connecting nodes that represent vehicles with overlapping charging time.  Simultaneously charging vehicles form cliques in the graph and it is sufficient to consider the maximal cliques in the graph.  It is an old result that interval graphs have at most linearly many maximal cliques all of which can be found in linear time \cite{shaohan1988maximal}.  Thus, the problem is in NP.

To show it is NP-complete, start by letting there be $K$ charging stations, each with capacity 1.  For each job $j$, create a vehicle that needs to charge from $s_j$ to $f_j$.  Let the cost $d_{vs}$ of pairing station $s$ with vehicle $v$ be 0 if $s \leq r_j$ and 1 otherwise.

There exists a feasible schedule to the HIS problem if and only if there exists a solution to the ILP of zero total cost.  Thus the short horizon problem is NP-complete.
\end{proof}

While solving the ILP is NP-complete, as with many problems in practice commercial solvers were able to solve sufficiently large instances in less than a tenth of a second.  Earlier it was argued that release dates for the long horizon problem can be computed so that this problem will always have a solution.  But what happens if under more aggressive parameterizations of $D$ there is no feasible solution in some iteration?  Since all vehicles that were present in the previous iteration could be assigned charge stations, we know the infeasibility of the problem must stem from vehicles that were just added to $V'$ and that have the latest charge times of all vehicles in $V'$.  One way to resolve this is by solving the short horizon problem without these vehicles, using a greedy assignment on the new vehicles, and then forcing any remaining unassigned vehicles to be rescheduled in the long horizon problem in the next iteration.  In simulations, this step was never needed.

In the real world, the simulated prediction of the mechanics of the vehicle could be incorrect and charge could be used faster than expected.  This could cause a vehicle to be unable to reach the charge station it was assigned previously because its charge level is now expected to be below zero at its assigned charging time.  Without compromising the practical applicability of this method, there are a variety of model properties that allow this problem to be addressed.

In this case, it is typically sufficient to make a hard constraint that any vehicle that was assigned to charge station $s$ in the previous assignment period must also have that station included in $S_v$ in the next iteration as well.  Recall from section \ref{sec:batteryconclusions} that for battery longevity reasons there is a buffer between what the operator calls $0$ charge and the actual empty level for the battery (assumption \ref{assumption:qmin}).  Thus a vehicle might have a charge level that is too low for any stations to be included in $S_v$ though in reality it has sufficient range to reach at least some stations.  All else failing, vehicles can travel to the nearest charge station and wait for the next available slot, even if it is outside of the assigned time window.  A variety of other resolutions could be used, such as letting the vehicle wait and be assigned a new charging time similar to the resolution proposed for aggressive parameterizations of $D$.

\subsection{Heuristic Long Horizon Planning}
\label{ss:objectives}

While the short horizon planning problem can be solved very efficiently in practice, the long horizon MILP is not practical for real-time use.  This is not surprising since it contains $O(VT)$ variables and constraints making it pseudo polynomial in size with the resolution of time used.  Time indexed formulations are known to be large and intractable in many settings.  While in the machine scheduling literature there exist exact techniques, such as column generation approaches \cite{van2000time}, for solving similar problems, these methods do not yield strong worst-case runtime guarantees and thus instead of adapting exact techniques for this novel MILP formulation we propose an efficient heuristic to update the charge schedule for vehicles.

As motivation for the heuristic, consider a system with a single vehicle and a single charge station.  The vehicle is always either operating or charging.  Supposing discharge and charge are linear function of time, the time for the vehicle to charge after driving for $t$ minutes can be modeled as $at + b$, where $b$ is the offline time spent driving to the charge station.  If the cost of operation is proportional to time offline for charging, then the average cost over the course of the day
\[
    g(t) = \frac{at + b}{at + b + t} = \frac{at + b}{(a + 1)t + b}
\]
is decreasing in $t$.  Meaning, the impact of being offline is minimized when the battery is used to the maximum extent and charging occurs as infrequently as possible. This directly motivates our heuristic:
\begin{heuristic}
    The vehicle with most remaining charge should be scheduled to charge as late as possible.  Other vehicles can be scheduled iteratively.
\end{heuristic}

In the online model we would like to only consider the next time each vehicle charges.  Even though vehicles must charge multiple times we hope that the objective of maximizing time until charging, as suggested in the example, will accomplish this goal.  However, choosing a charging schedule that waits as long as possible has a downside - it maximizes the amount of work that needs to be done in the single charge time horizon we consider.

It is well known in the literature of time dependent job scheduling that processing the shortest job first minimizes the makespan (total work) of the schedule.  In fact, the problem closely resembles the time dependent jobs with deadlines problem studied by \cite{cheng2000single}.  However, by scheduling vehicles with the most remaining charge last we are doing the opposite - we are placing the shortest jobs at the end, resulting in a final schedule where realized processing times are not in any sorted order.

We believe the justification for our heuristic stands even in light of this.  The total amount of work to be done recharging is first the fixed cost of detouring to a charge station and second charging time, which we assume is proportional to the miles the vehicle have driven.  Assuming the vehicle drive the same number of miles under any scheduling scheme, the best schedule will be one the aims to minimize the number of times vehicles charge.  This is what our heuristic does.

This conflict between minimizing work in a single scheduling horizon versus minimizing work over the entire day also illustrates why it was important for the long horizon MILP to produce schedules for the entire day rather than just find the next time each vehicle charges.  Under the shortest work rule used in time dependent job scheduling (in its strictest sense) vehicles that finish charging would have to immediately turn around and charge again while vehicle with depleted batteries would never be charged!

\subsubsection{Algorithm.}

In order to implement the planning heuristic we will first define some objects to store static data and the dynamically updated state of slack in the problem's constraints and then describe how they are used in the algorithm.

We start by creating an object $\mathtt{Constraints}$ which can be queried for each time of day to know the slack on capacity constraints (total capacity $K$) and availability constraints (given by $R(t)$).  As decisions are made, changes can be written to the $\cons$ to show reduction in slack.

The basic algorithm works as follows.  For each vehicle, we compute the value of $e_v$ and store these values in the discrete map $\earliest$.  Similar to $A_{v, t}$ in the MILP formulation, for each previous assignment, we update our $\cons$ object to reflect reduction in availability so that fewer jobs are scheduled at that time.

Since we can compute the charge in vehicles $Q_{v, 0}$ at time $e_v$, we can use $q_{est}$ to extrapolate the first time each vehicle is expected to run out of charge.  This is the deadline for the vehicle, which we store in the discrete map $\deadline$.

Next, we sort the list of vehicles according to their $\deadline$.  Vehicles with later deadlines are given the highest priority, they will be scheduled as far in the future as possible.  We store this sorted list as $\prior$.

Using this priority based scheduling, sometimes it is not possible to schedule all the vehicles to charge given capacity, availability, and deadline constraints.  When this happens, all elements of the schedule are pushed into the future until sufficient space is made, using algorithm (\ref{alg:pushback}) $\mathtt{PushBack}$.  This operation is done by first pushing back vehicles with lowest priority - minimizing the maximum extent any vehicle is moved back compared to its current assignment.  This operation can be performed efficiently using doubly linked lists to avoid redundant computations for groups of vehicles scheduled to charge at the same time.

This could cause some vehicles to have charge times after they would be expected to run out of charge.  Since $q_{est}$ is a slightly pessimistic estimate of charge usage, it is often the case that the vehicles that are pushed back will not actually run out of power before charging.  In addition, since we assume that the 0\% mark is actually a safe buffer $q_{\min}$ away from absolute zero this will not be a significant issue for the vehicle.  In any case, if this poses a practical issue for the operator modifications can easily be made to passenger assignment and vehicle routing algorithm to halt a vehicle at a nearby charge station if it is nearing empty charge.

The $\mathtt{PushBack}$ algorithm is shown in algorithm \ref{alg:pushback} and the full algorithm for charge scheduling is given in algorithm \ref{alg:quick_long}.  The $\mathtt{PushBack}$ algorithm uses two special functions: $\mathtt{ClearAvailability}$ and $\mathtt{ClearCapacity}$.  Each of these takes as an input the working copy of $\cons$, the $\sched$, the vehicle's priority, and time for which the vehicle to be scheduled is blocked either by availability or capacity constraints.  For ease of notation, the $\cons$ and $\sched$ arguments are omitted in the algorithms.

$\mathtt{ClearAvailability}$ and $\mathtt{ClearCapacity}$ do the following: if there is a set of vehicles of higher priority that can be removed that will clear sufficient capacity or availability, the smallest set of such vehicles is returned.  Keeping the recursive nature of the $\mathtt{PushBack}$ algorithm in mind, the functions always return the vehicle of highest priority among all that would clear the constraints in order to minimize work in future iterations.

\begin{algorithm}
    \caption{$\mathtt{PushBack}(\naturals \to (V, T) \ \mathtt{Items}, \sched, \cons)$}
    \label{alg:pushback}
    \begin{algorithmic}[1]
        \If {$\mathtt{Items} = \emptyset$}
            \State \Return $\sched$.
        \Else
            \State $(v, t) = \mathtt{Items}$.pop() (lowest priority item).
            \For {time $s = t, \dots, T$}
                \State $\cons' = \cons$.
                \State $\sched' = \sched$.
                \State $\mathtt{BumpList} : \naturals \to (V, T) = \{\}$.
                \For {increment $i = 0, \dots, $NeededCharge}
                    \If {$\cons'$ availability insufficient at time $(s + i)$}
                        \State obstruction $[v'] =
                                \mathtt{ClearAvailability}(\prior_v, s + i)$.
                        \If {no obstruction}
                            \State Continue to iteration $s = s + 1$.
                        \Else
                            \For {$v' \in [v']$}
                                \State Unschedule $v'$ in $\sched'$.
                                \State Update $\cons'$, $\sched'$.
                                \State $\mathtt{BumpList} \cup= 
                                        \{\prior_{v'} : (v', s + i)\}$.
                            \EndFor
                        \EndIf
                    \EndIf
                    \If {$\cons'$ capacity insufficient at time $(s + i)$}
                        \State obstruction $[v'] = \mathtt{ClearCapacity}(\prior_v, s + i)$.
                        \If {no obstruction}
                            \State Continue to iteration $s = s + 1$.
                        \Else
                            \For{$v' \in [v']$}
                                \State Unschedule $v'$ in $\sched'$.
                                \State Update $\cons'$, $\sched'$.
                                \State $\mathtt{BumpList} \cup= 
                                        \{\prior_{v'} : (v', s + i)\}$.
                            \EndFor
                        \EndIf
                    \EndIf
                \EndFor
                \State Schedule vehicle $v$ at time $(s + i)$ in $\sched'$.
                \State Update $\sched'$ and $\cons'$.
                \State \Return $\mathtt{PushBack}(\mathtt{Items} \cup \mathtt{BumpList},
                        \sched', \cons')$.
            \EndFor
        \EndIf
    \end{algorithmic}
\end{algorithm}

\begin{lemma}
The complexity of the $\pushback$ algorithm (algorithm \ref{alg:pushback}) is $O(T)$.
\end{lemma}
\begin{proof}
Note that the first loop runs over $O(T)$ time frames.  Assume that the maximum number of time windows a vehicle may need for a full charge, NeededCharge, are $O(1)$.  Within each loop, we have a constant number of steps.  The $\mathtt{ClearAvailability}$ and $\mathtt{ClearCapacity}$ functions take constant time by looking up the time index in an array and then inspecting at most the first $c$ vehicles assigned at that time, with $c$ being the number of time periods needed for charging and assumed to be $O(1)$.  Since the $\mathtt{ClearAvailability}$ and $\mathtt{ClearCapacity}$ function prefer to select vehicles with the highest priority among all vehicle eligible at each point in time, running the Unschedule and Schedule routines takes constant time by storing elements for each vehicle at each time in a doubly linked list structure.  Updating constraints is also constant time since we have $O(1)$ removals to clear space and $O(1)$ insertions to add the current vehicle.

To see that the algorithm is $O(T)$ including recursion, notice that when vehicles are removed they are only considered for rescheduling at a later time.  This is because we always choose to remove the vehicle scheduled at that time with the highest priority and we can only bump vehicles of higher priority when choosing a new charge time for the vehicle.  Thus each time period is considered at most a constant number of times and the recursive algorithm is $O(T)$.
\end{proof}

\begin{theorem}
The complexity of the Long Horizon Heuristic (algorithm \ref{alg:quick_long}) is 
\[ O(n\log(n) + nT). \]
\end{theorem}

\begin{proof}
The algorithm sorts all vehicles by priority using $O(n\log(n))$ time.  Then the algorithm tries to schedule the vehicle in one of $O(T)$ time periods with each period requiring a constant amount of work to verify whether there is space or not.  Insertion into the schedule takes constant time if the duration is $O(1)$, which we assume, and the schedule is implemented as an doubly linked list.  

If the $\pushback$ algorithm is required, this adds at most $O(T)$ time to the loop.  Thus the loop runs in time $O(nT)$.  Since the algorithm begins with sorting, the total running time is $O(n\log(n) + nT)$.
\end{proof}

\begin{algorithm}
  \caption {Long Horizon Heuristic$(V, \cons, \deadline, \earliest, \prior)$}
  \label{alg:quick_long}
  \begin{algorithmic}[1]
    \State $\sched : V \to [T] = \{\}$.
    \For {each vehicle $v \in \prior$ in descending order}
        \For {each time $s$ from $\deadline_v$ down to $\earliest_v$}
            \State $\text{duration} =$ number of time windows $v$ needs if start charging at $s$.
            \If {Can schedule $v$ at time $s$}
                \State $\sched = \sched \cup \{v : [s, \dots, s + \text{duration}]\}$.
                \State break
            \EndIf
        \EndFor
        \If {not yet scheduled}
            \State $vs : \naturals \to (V, T) = \{\prior_v : (v, \earliest_v)\}$.
            \State $\sched = \mathtt{PushBack}(vs, \sched, \cons)$.
        \EndIf
    \EndFor
  \end{algorithmic}
\end{algorithm}

\section {Numerical Results}
\label{sec:results}

To test the performance of our proposed methods, simulations were performed using publicly available taxi ridership data in Manhattan from August 2013~\cite{nycdata}.  We use a threshold of $T_{SL} = 45$ minutes and let requests arrive in batches every $T^B = 60$ seconds.  We weighted the objective function to first maximize the number of requests served and then, second, minimize total vehicle travel time.

The charging rate $\eta$ was chosen so that vehicles would take 30 minutes for a full recharge.  The estimated discharge rate over time $q_{est}$ was estimated using distance over time simulation results from ICE simulations whose results were dependent on the choice of request data set used (e.g., which day).  Recall, discharge does not have to be computed as a function of distance but we choose to do so for simplicity.

The simulations were performed on a desktop using an Intel Core i7-6700 CPU with 8 cores and 16GB of RAM.  The simulation was written in C++ and optimization was performed using Mosek version 8.

The experiments tested four different methods:
\begin{enumerate}
    \item  MILP:  Long horizon MILP model and short horizon MILP model.
    \item  Heuristic:  Long horizon heuristic in combination with the short horizon MILP model.
    \item  Benchmark:  Simple algorithm that does not use information about future demand, defined in section \ref{subsec:benchmark}.
    \item  ICE Baseline:  None of the vehicles are EV.
\end{enumerate}

In the rest of this section we discuss how the inputs to the simulations were produced and discuss the results.

\subsection {Charge Station Location}
\label{chargestationlocation}

As discussed in the literature review, a variety of approaches have been proposed to choose locations for charging stations.  Following the spirit of the methods that use K-means problems, we considered locating the charge stations using a weighted k-means clustering algorithm.  As a proxy for determining where the vehicles would have no passengers, and thus available to charge, we referred to the taxi data set which consists of historical request data for Manhattan in the form of origin-destination pairs.  Every time a vehicle charges it is between its last drop-off and next pickup and so this data set can be interpreted as samples of possible locations at which vehicles will be empty.

Combining the two ideas, we performed weighted k-means clustering where weights were proportional to the number of originations/destinations at each node and normalized so the total weight of the graph was equal to the desired total charging capacity.  The central node of each cluster was chosen to be the charge station and the total capacity of the charging infrastructure was distributed in proportion to the corresponding cluster weights.

\subsection{Benchmark Method}
\label{subsec:benchmark}

As a benchmark, we created a simple algorithm to schedule vehicles that does not take into account any information about future demand.  The method is similar to the method used in \cite{fehn2019modeling} which greedily assigns vehicles to the closest charge station, though we additionally allow the vehicle to go to further stations if it is able to charge sooner there.  In addition, we allow charge going below a threshold to be the signal for charging similar to \cite{fehn2019modeling}, though rather than implementing this as a check whenever a passenger is dropped off we also check before assigning a new passenger to a vehicle, a necessary change since \cite{fehn2019modeling} only considers single capacity vehicles.

In this algorithm, vehicles make individual decisions of when to charge.  Vehicles operate all day until their charge goes below the threshold $q_{\min}$.  Then the vehicles continue to serve requests that have been assigned to them, but refuse to accept any new passengers.  Once the final passenger has been dropped off, the vehicle requests space at a charge station that will allow the vehicle to charge the soonest.  To help keep vehicles near areas of high demand, we limit the set of charge stations the vehicle can charge at to be within a specified time radius of its location, which we choose to be 15 minutes.  Once the vehicle is finished charging it returns to service.

This is described in Algorithm \ref{alg:benchmark}.  Additionally, the Algorithm \ref{alg:naive} describes how a vehicle that is empty and is below the charge threshold is assigned to a charging station in a greedy manner.

\begin{algorithm}
    \caption {Benchmark Charging Algorithm}
    \label{alg:benchmark}
    \begin{algorithmic}[1]
        \For {each vehicle $v$}
            \If {charge in $v$ is below threshold}
                \State Forbid vehicle from accepting new passengers in next assignment iteration.
                \If {vehicle is empty and has no charge station assignment}
                    \State Assign charge station using Greedy Assignment algorithm.
                \EndIf
            \Else
                \State  Clear vehicle's charge station assignment.
            \EndIf
        \EndFor
    \end{algorithmic}
\end{algorithm}

\begin{algorithm}
  \caption {Greedy Station Assignment(\textbf{vehicle})} 
  \label{alg:naive}
  \begin{algorithmic}[1]
    \State $\ell \gets $ current location of \textbf{vehicle}.
    \State $S \gets $ set of charging stations within distance threshold.
    \State Times $\gets \{\text{time}(\textbf{vehicle}, s) \text{ for } s \in \text{Stations}\}$.
    \State Earliest $\gets \{(s, \text{first available slot time at $s$} \text{ for }) s \in S\}$.
    \State Assignment $\gets \arg\min_{s \in S} \max\{\text{Times}_s, \text{Earliest}_s\}$.
    \State Update $S$ with Assignment.
    \State \Return $s$, $\max\{\text{Times}_s, \text{Earliest}_s\}$.
  \end{algorithmic}
\end{algorithm}

\subsection {Producing Availability Requirements $R(t)$}
\label{subsec:availability}

To compute the vehicle availability requirement vector $R(t)$ for methods that are aware of information about future demand, we employed a simple procedure that roughly models the required number of vehicles by time of day as a weighted average of the entire fleet and a scaled multiple of the customer demand profile.  Letting $d(t)$ be an expected demand profile vector and $\lambda$ be a scalar parameter, our model for $R(t)$ is described by the equation
\[ R(t) = |V| \Big(\lambda d(t) + (1 - \lambda)\Big). \]

First, to select $d(t)$ we started by considering all requests in the input data file and associated each with an interval starting when the request is made and with duration equal to the travel time from the request's origin to destination.  For every thirty minute block of time during the day we counted the number of requests whose intervals overlap any portion of that block.  We used the results to form $d(t)$ as a piecewise constant function and normalized so the largest value was 1.

Next, to select $\lambda$ we started by computing an estimate for the number of hours of charging that must be done given the size of the fleet and the range of the vehicles.  We first tried using a $\lambda$ that made the number of offline vehicle hours equal to the total charging demand.  This turned out to not be sufficient since vehicles have limitations on when charging can be done, such as not running out of power and being limited to 100\% charge.  The peaks and troughs in the demand function were sometimes too spread out or too heavily concentrated in the morning to yield an effective requirement function.  Thus we decided another method was needed to determine a reasonable $\lambda$.

We used a greedy algorithm to create a naive charging schedule for the day given the expected battery discharge rate and $\lambda = 0$, so no feasible solution could be found.  Then $\lambda$ was continuously raised until it was possible to find a feasible schedule using the greedy algorithm. This $\lambda$ was then used for the experiments.  Since this method is only a heuristic, it is possible that optimal solutions may exist for tighter (smaller) $\lambda$, meaning more vehicles required on the road at all times.  In practice an operator would have many factors to consider beyond the scope of this paper and so we leave open the space of possible procedures for constructing $R(t)$.  We explored some of the possibilities of other constructions in $R(t)$ in our experiments as well, see tables \ref{table:220_rt_sensitivity} and \ref{table:1000_rt_sensitivity}.

\subsection{Fleet Size and Reduced-Size Map for MILP Method Experiments}
\label{fleetsize}

Simulations with ICE vehicles and those using computationally efficient methods, such as the benchmark and long horizon heuristic methods, were tested with up to 2000 vehicles.  Since the long horizon MILP models are large and do not scale well in simulations with many vehicles, in order to run experiments with the MILP method we reduced the number of vehicles and produced a reduced size graph for testing.  

To make the reduced-size graph we selected a portion of Manhattan roughly from 19th Street to 59th street.  This reduced the total number of intersections in the graph by about 78\%. 
All requests that did not originate and terminate in this region were removed from the input data set.  Calibrating the experimental service rate by randomly filtering out requests was put off until the number of vehicles had been chosen.

To determine the number of vehicles to use in the reduced size experiment, we ran a small number of iterations of the long horizon MILP at various vehicle counts.  We chose 220 vehicles since it was near the limit of how many vehicles the MILP solver could handle in a reasonable period of time.  Based on observations, we choose two time limits for the MILP solver to use in the final experiments of 180 and 55 seconds, after which in each case we use the best solution found so far.  Giving the MILP solver more than 180 seconds did not result in significantly lower objective values.  To select the demand level, we ran the ICE method on the reduced size system at various levels of demand with 220 vehicles and reduced the number of requests until the service rate reached 90\%.  This resulted in filtering 2/3 of the requests from the input data set.

In summary, we chose 220 vehicles to use on the reduced-size map on which we perform experiments on all four methods.  In addition, we ran experiments using 1000 and 2000 vehicles for the heuristic method, benchmark method, and ICE baseline.

\subsection{Battery Range and Expected Charge Duration}
\label{estimatedbattery}

While battery ranges on small personal vehicles have grown significantly in recent years, albeit with high price tags, ranges for larger commercial vehicles lag behind.  Since our target system is a high-capacity ridepool system, we surveyed information about charge range for electric service and passenger vans.  When reviewing the data, we assume that tests performed using NEDC have real-world energy consumption that is 38\% higher \cite{mock2014laboratory} and tests that do not list a method were performed using the newer WLTP standard with resulting real-world energy consumption that is 14\% higher \cite{dornoff2020way}.  By reviewing a list of electric vans that are to be available in 2020 \cite{drivingelectric}, we find that the average listed electric van has a range of 160 km using the default or lower capacity battery option and 200 km using the default or higher capacity battery option.  From this, we selected 180 km as the range for our hypothetical vehicles in simulation.

Both the MILP method and the heuristic for charge scheduling require an estimate of when batteries in vehicle will need to be charged.  However, the assumption that energy usage is linear in distance traveled requires us to produce a separate estimate for how long batter charges last.  For each of the 3 fleet sizes we first run a simulation using ICE vehicle.  Then, using the number of miles driven during the simulated day we produce an estimate of the number of hours until discharged.  For 220 vehicles, the result was 13 hours and for 1000 and 2000 vehicles it was 10 hours, different values due to differences in request density and map size discussed in \ref{fleetsize}.  See Table \ref{table:battery_estimates} to see how sensitive these estimates were.

\subsection{Data Processing and Simulator Implementation}

The Manhattan taxi ridership data \cite{nycdata} was converted into a format usable by the simulator by map matching the origin and destination coordinates onto the road network used in the simulator.  The time that each request was submitted in the ridership data was preserved and the resulting combinations of time and locations were stored in a file.  To run the simulation, the file is loaded and each request is presented to the assignment algorithm at the end of the batch interval during which it actually arrived in the ridership data.

At the end of each batch of processing, the central operator gives each vehicle an updated assignment list.  The operator can update the assignment list at any time without regard for the state of the vehicle, with the exception that the assignment must be feasible given the quality of service constraints and passengers that are already onboard must remain in the assignment list until the vehicle drops them off.  After each batch of assignments, the simulator executes the following steps for each vehicle:
\begin{enumerate}
    \item  From the trip assigned from the RTV graph, get optimal ordering of passenger pickup and drop off, given current vehicle position, onboard passengers, time remaining for quality of service constraints for each request, and any electric constraints on the vehicle.
    \item  Expand the ordering with an edge-by-edge path.
    \item  Follow the path, making pickups and drop offs as appropriate.  Pause when the end of the batching interval is reached, even if the vehicle has not completed its assignment or has only partially traversed the current edge.  Travel will resume after the next batch of assignment updates have been processed.
\end{enumerate}

\subsection {Results}

We primarily evaluated the performance of each method by the percentage of requests served, though we also report vehicle miles traveled as well as other statistics.  The statistics we report are similar to those used by \cite{alonso2017demand}, \cite{ma2014real}, and \cite{ota2016stars}.  A summary of all statistics used and their abbreviations are given below.  See Table \ref{table:statistics} for a more compact reference.

\begin{itemize}
    \item  \emph{Service rate}, abbreviated as \emph{Rate}, is the ratio of the number of requests that were served to the total number of requests that were made.
    \item  \emph{Waiting time}, denoted \emph{WT}, is the average amount of time, among requests that were served, between system entry time and vehicle pickup.
    \item  \emph{Riding time}, denoted \emph{RT}, is the average amount of time between requests boarding vehicles and alighting from them.
    \item  \emph{Total delay}, denoted \emph{delay}, is the average difference between shortest path travel time from requests' origins to destinations and the actual time from request system entry to drop-off.
    \item  \emph{Absolute utilization}, denoted \emph{Abs}, is the ratio of the total number of rider-minutes in the system to the total number of vehicle-operating minutes in the system.
    \item  \emph{Rider share rate}, denoted \emph{rider}, is the ratio of the total number of rider-minutes to the total number of vehicle minutes while at least one rider was in the vehicle.
    \item  The \emph{shared rate}, denoted, \emph{shared}, is the percentage of riders who shared their ride with another rider at any point in their journey.
   \item  \emph{Distance} is the total number of kilometers driven by the entire fleet of vehicles over the course of the day.
\end{itemize}

\begin{table}[b!]
    \centering
    \caption{Description and abbreviations for statistics reported from experiment.}
    \begin{tabular}{l|l|l}
        Statistic & Abbr. & Description \\
        Service Rate & Rate & Requests served / requests made \\
        Waiting Time & WT & Average of pickup time minus request time \\
        Riding time & RT & Average time in vehicle \\
        Total delay & delay & Average ride time minus direct travel time \\
        Abs utilization & Abs & Rider minutes / total vehicle minutes \\
        Rider share rate & rider & Rider minutes / non-empty vehicle minutes \\
        Shared rate & shared & Percent of riders that shared a ride \\
        Distance & distance & Average distance vehicles travel \\
    \end{tabular}
    \label{table:statistics}
\end{table}

The results for experiments with 220 vehicles are shown in Table \ref{table:220_general_results}.  Compared to the simulation with internal combustion engine vehicles (ICE), the benchmark method has the lowest service rate, which performs about 8\% worse.  All of the non trivial methods for electric vehicles perform well, recovering around 80\% of the lost service rate compared to the upper bounding ICE vehicle simulation.

\begin{table}[b!]
    \caption{Simulation results for 220 Vehicles.  The MILP and Heuristic methods had significantly higher services rates than the benchmark.  The ICE baseline has no charging and is an upper bound on service rate performance.}
    \begin{filecontents*}{tables/1-220-vehicles.csv}
Method,ILP Time Limit,Rate,WT,RT,Delay,Abs,Rider,Shared,Distance
ICE,,90.99,125.96,534.68,86.77,1.35,1.63,92.06,256.68
MILP,(180 seconds),89.19,125.93,540.15,93.23,1.34,1.67,92.99,248.85
MILP,(55 seconds),89.37,126.71,540.74,93.62,1.34,1.67,93.17,249.33
Heuristic,,89.35,126.88,541.47,94.16,1.34,1.67,93.21,248.73
Benchmark,,82.56,125.3,532.56,88.39,1.22,1.64,91.83,233.22
\end{filecontents*}
\csvautotabular{tables/1-220-vehicles.csv}
    \label{table:220_general_results}
\end{table}

The reason for the poor performance of the benchmark method is largely due to the offline time of the vehicles.  Figure (\ref{figure:220_naive}) shows that the vehicles charge in a single short period.  Although the number of vehicles charging is capped at the number of charging stations in the simulation, the drop of the vehicle miles traveled (shown in orange) indicates congestion at the charge stations and that many vehicles are waiting offline.  Vehicles wait an average of 76 minutes before they can begin charging due to vehicles running out of charge at similar times.

\begin{figure}[h!t]
  \centering
  \includegraphics[scale=0.6]{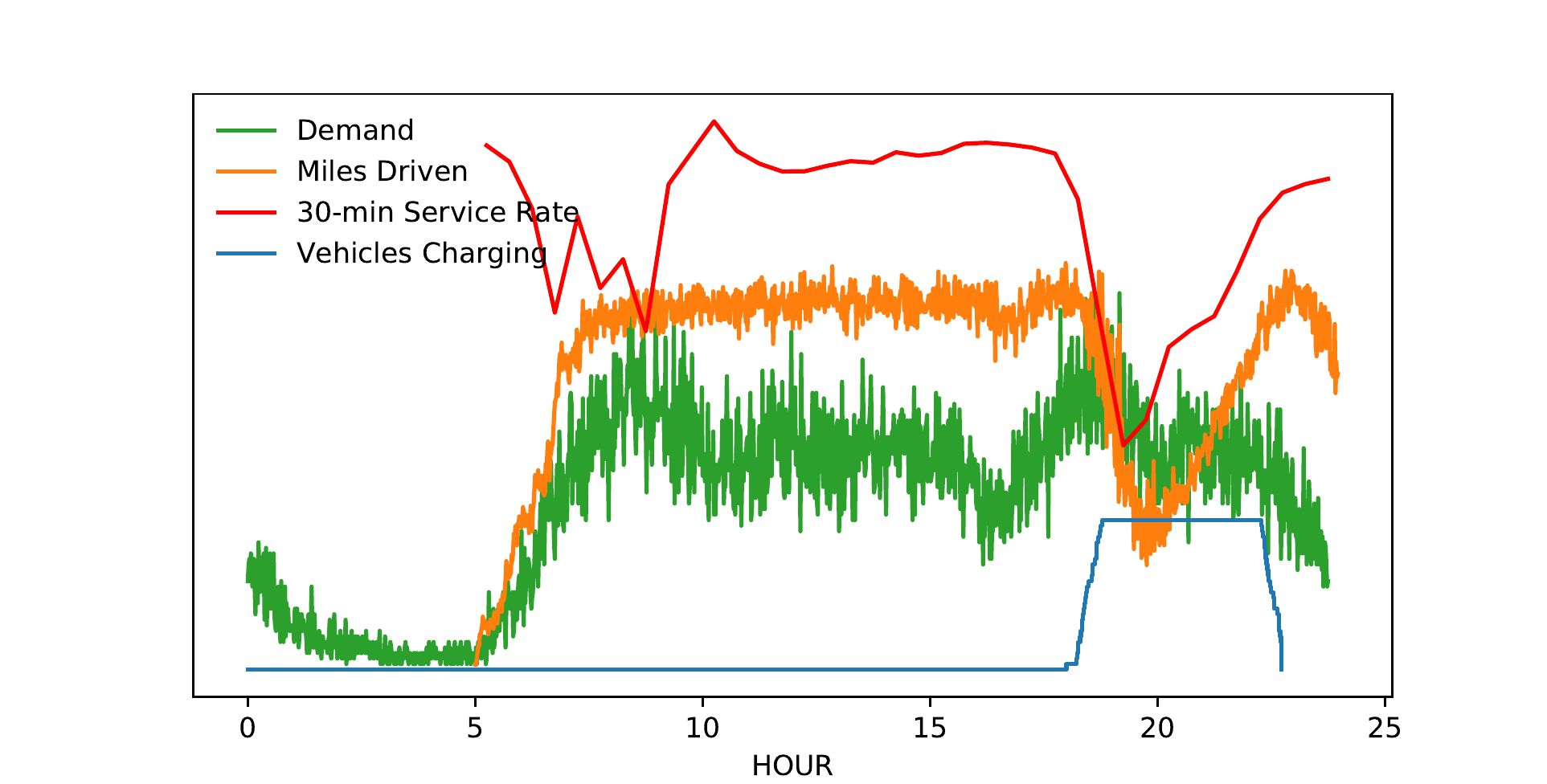}
  \caption{Benchmark method for 220 vehicles.  Notice that vehicle charging, in blue, is concentrated during a single period in late afternoon.  The number of vehicle miles traveled during this time, in orange, plummets.  The 30 minute service rate during the period, shown in red, also dips resulting in a reduction of the total cumulative service rate.} 
  \label{figure:220_naive}
\end{figure}

The variation in performance of the heuristic and MILP methods comes down to the quality of the availability vector $A(t)$ given as an input to the problem.  We see that the MILP and heuristic methods perform similarly.  In Figure (\ref{figure:220_ordering}), we can see that avoiding congestion at charge stations played a major role in reducing the negative impact on service rate.

\begin{figure}[h!t]
  \centering
  \includegraphics[scale=0.6]{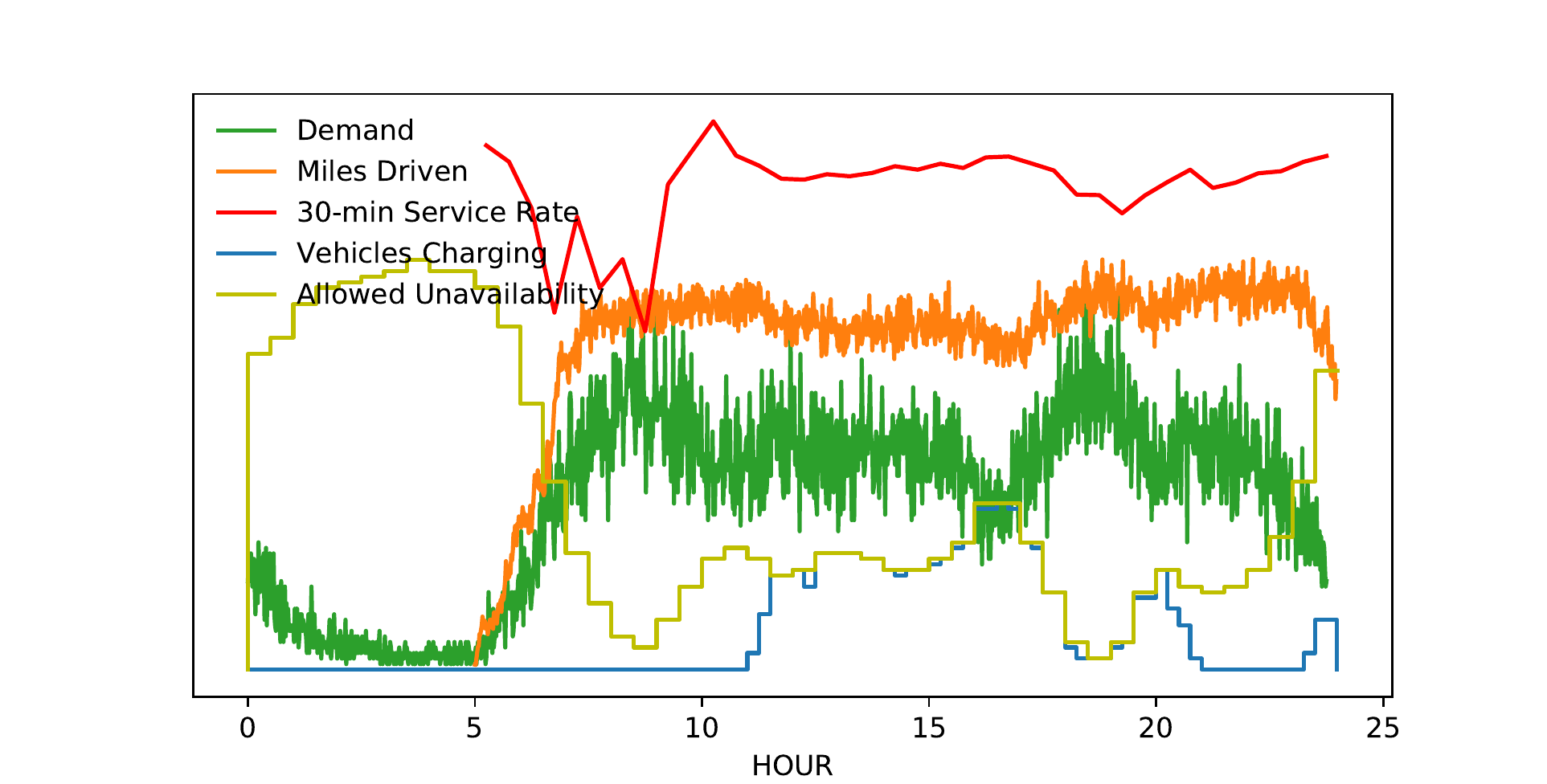}
  \caption{Heuristic method on 220 vehicles.  The fractional availability used, shown in blue, is defined to include both charging and the 15 minute period prior to charging in which it is assumed that vehicles lack full ability to serve customers due to constraints on reaching the charging station.  The maximum allowed offline availability, shown in yellow, indicates the limit by time of day of unavailable vehicles.}
  \label{figure:220_ordering}
\end{figure}

We can only assume the heuristic method preforms similarly to the MILP methods in larger cases since the long horizon MILP cannot be solved in reasonable time.

\begin{table}
\centering
\caption{Results for simulations with 1000 Vehicles.  Compared to the benchmark, the heuristic method mitigated the loss in service rate compared to the upper bounding ICE baseline.}
\label{table:1000_results}
\begin{filecontents*}{tables/2-1000-vehicles.csv}
Method,Rate,WT,RT,Delay,Abs,Rider,Shared,Distance
ICE,88.44,176.87,793.97,190.46,2.1,2.37,98.19,342.59
Heuristic,86.38,177.71,796.34,196.81,2.06,2.43,98.17,327.51
Benchmark,83.35,175.79,791.68,196.61,1.98,2.44,98.22,315.77
\end{filecontents*}
\csvautotabular{tables/2-1000-vehicles.csv}
\end{table}

\begin{table}
\centering
\caption{Results for simulations with 2000 Vehicles.  As with the experiments with fewer vehicles, the heuristic method has a service rate much closer to the upper bounding ICE than the benchmark algorithm.}
\label{table:2000_results}
\begin{filecontents*}{tables/EXTRA-10-2000-vehicles.csv}
Method,Rate,WT,RT,Delay,Abs,Rider,Shared,Distance
ICE,93.11,171.38,781.63,174.22,2.18,2.56,98.25,331.68
Heuristic,92.51,172.71,785.88,179.93,2.18,2.61,98.21,323.7
Benchmark,90.08,171.01,781.02,180.7,2.11,2.63,98.34,313.25
\end{filecontents*}
\csvautotabular{tables/EXTRA-10-2000-vehicles.csv}
\end{table}

The results of the simulation for 1000 vehicles are shown in Table \ref{table:1000_results}.  Here again we see the heuristic method outperform the benchmark method and recover more than 50\% of the lost service rate as compared to the ICE vehicles simulation.  Vehicles in the benchmark method spent an average of 60 minutes waiting for an available charging station.  Figures (\ref{figure:1000_naive}) and (\ref{figure:1000_ordering}) show that a similar concentration of charging demand affected the benchmark method again, similar to the case with 220 vehicles.  Results for simulations with 2000 vehicles, shown in Table \ref{table:2000_results}, show the heuristic method similarly performing well with respect to the gap between the ICE method and benchmark.

\begin{figure}[h!t]
  \centering
  \includegraphics[scale=0.6]{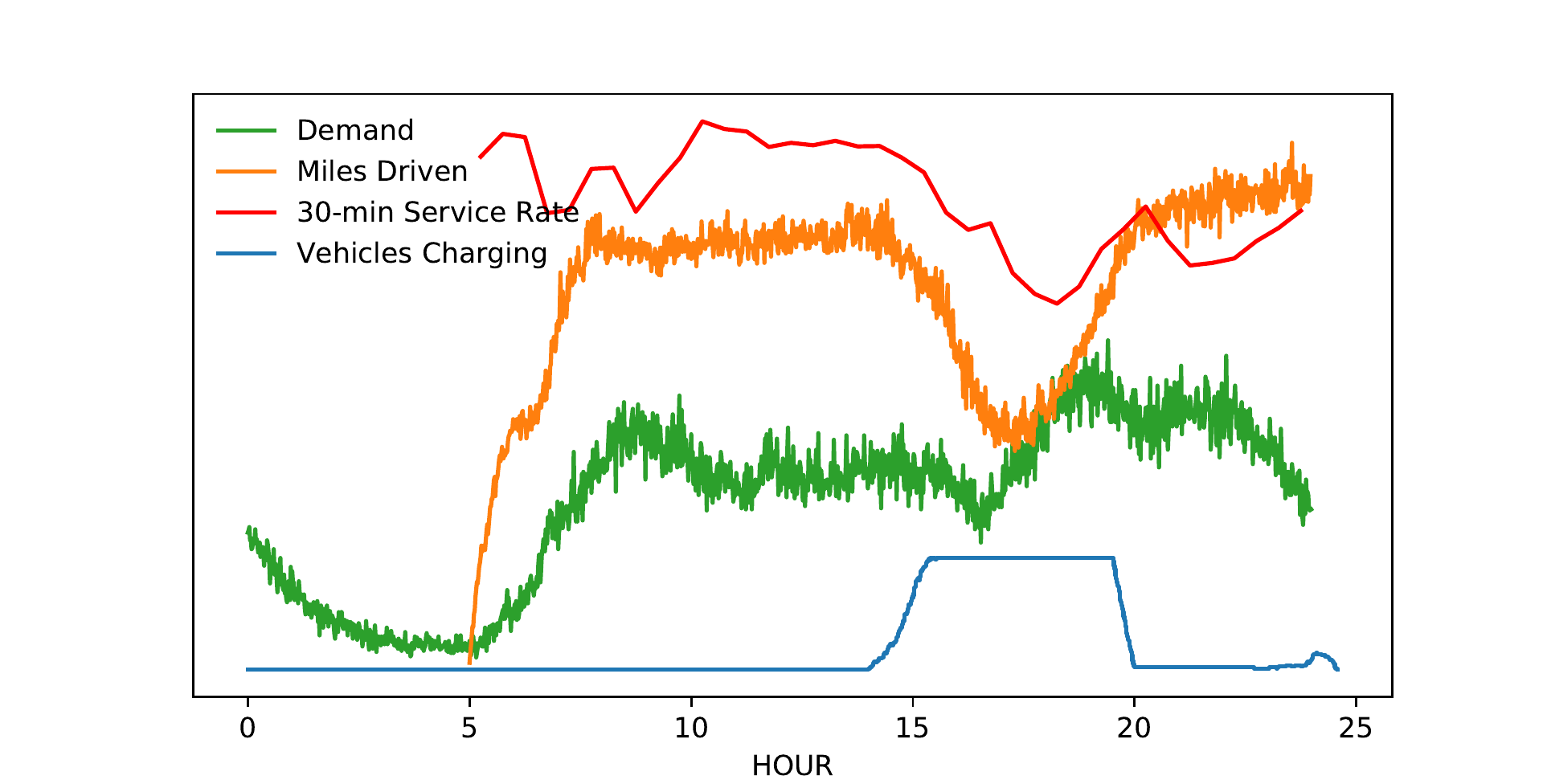}
  \caption{Benchmark method on 1000 vehicles.  This shows that the the vehicles all needed to charge at a similar time.  Some of the effect is hidden since only charging vehicles are counted in blue.  A larger number of vehicles are offline waiting to be charged during that time, only visible through the great reduction in vehicle miles traveled shown in orange.} 
  \label{figure:1000_naive}
\end{figure}

\begin{figure}[h!t]
  \centering
  \includegraphics[scale=0.6]{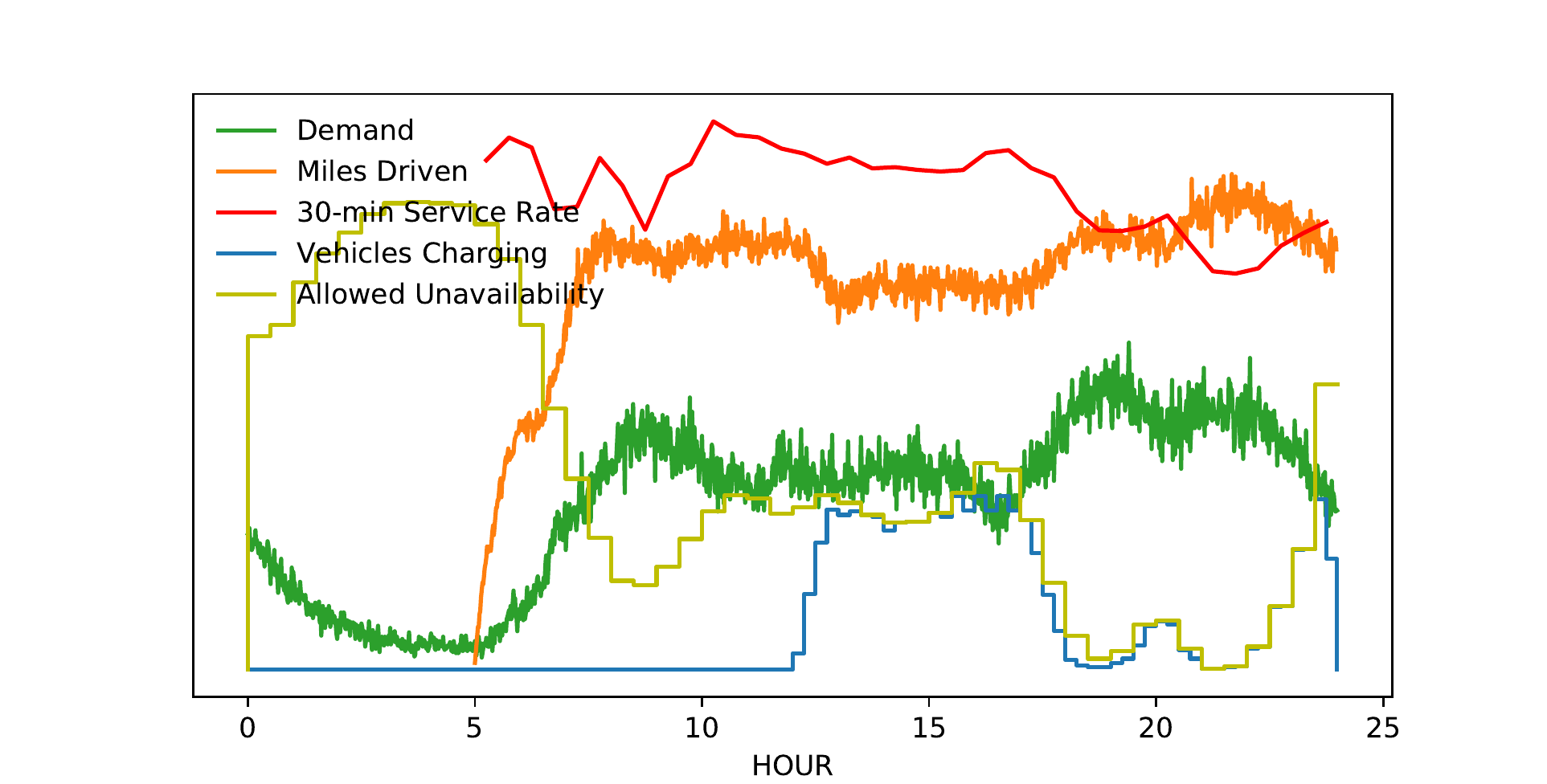}
  \caption{Heuristic method on 1000 vehicles.  This shows that the times the vehicles charged at were better spread, as evidenced by the small reduction in vehicles miles shown in orange.}
  \label{figure:1000_ordering}
\end{figure}

There are situations, however, where the performance of the benchmark method is comparable to that of the heuristic method.  As an example, consider again the simulation with 1000 vehicles and suppose they are initially fully charged at midnight.  The low levels of demand in the early morning do not require all vehicles be used and so the distribution of charges among vehicles is widely spread by 7am.  This has the effect of preventing vehicles from needing to charge at the same time.  In fact, on average vehicles only wait 5 and a half minutes before getting a charging assignment.  Figure (\ref{figure:1000_naive_midnight}) shows this.  Table \ref{table:betternaive} shows the results of some 220 and 1000 vehicle simulations starting at midnight.  But in practice we expect the fleet to charge overnight for the early morning and thus the simulations starting at midnight are likely not the best representative of an actual system.

In this particular case for 1000 vehicles, not only were the charging times well spread by a fortunate coincidence, most of the vehicles were in service during the peak demand period with charging levels slowly resuming as demand falls at the end of the day.

\begin{table}
\caption{Vehicle simulations from midnight.  When experiments start at midnight, the low demand in the early morning causes the distributions of charges in vehicles to spread.  As a result, the service rate of the benchmark method is similar to that of the heuristic method.}
\label{table:betternaive}
\centering
\begin{filecontents*}{tables/3-midnight.csv}
Veh. Count,Method,Rate,WT,RT,Delay,Abs,Rider,Shared,Distance
220,ICE,91.06,124.92,532.58,83.93,1.11,1.59,89.98,272.93
220,Heuristic,89.59,125.95,538.75,91.0,1.1,1.63,91.25,264.23
220,Benchmark,86.79,125.25,533.4,86.34,1.06,1.6,90.51,259.96
1000,ICE,88.01,177.05,795.88,189.11,1.78,2.33,97.83,381.81
1000,Heuristic,85.73,178.23,800.15,197.43,1.74,2.39,97.86,364.01
1000,Benchmark,85.94,177.68,801.98,198.08,1.75,2.39,98.03,367.71
\end{filecontents*}
\csvautotabular{tables/3-midnight.csv}
\end{table}

\begin{figure}[h!t]
  \centering
  \includegraphics[scale=0.6]{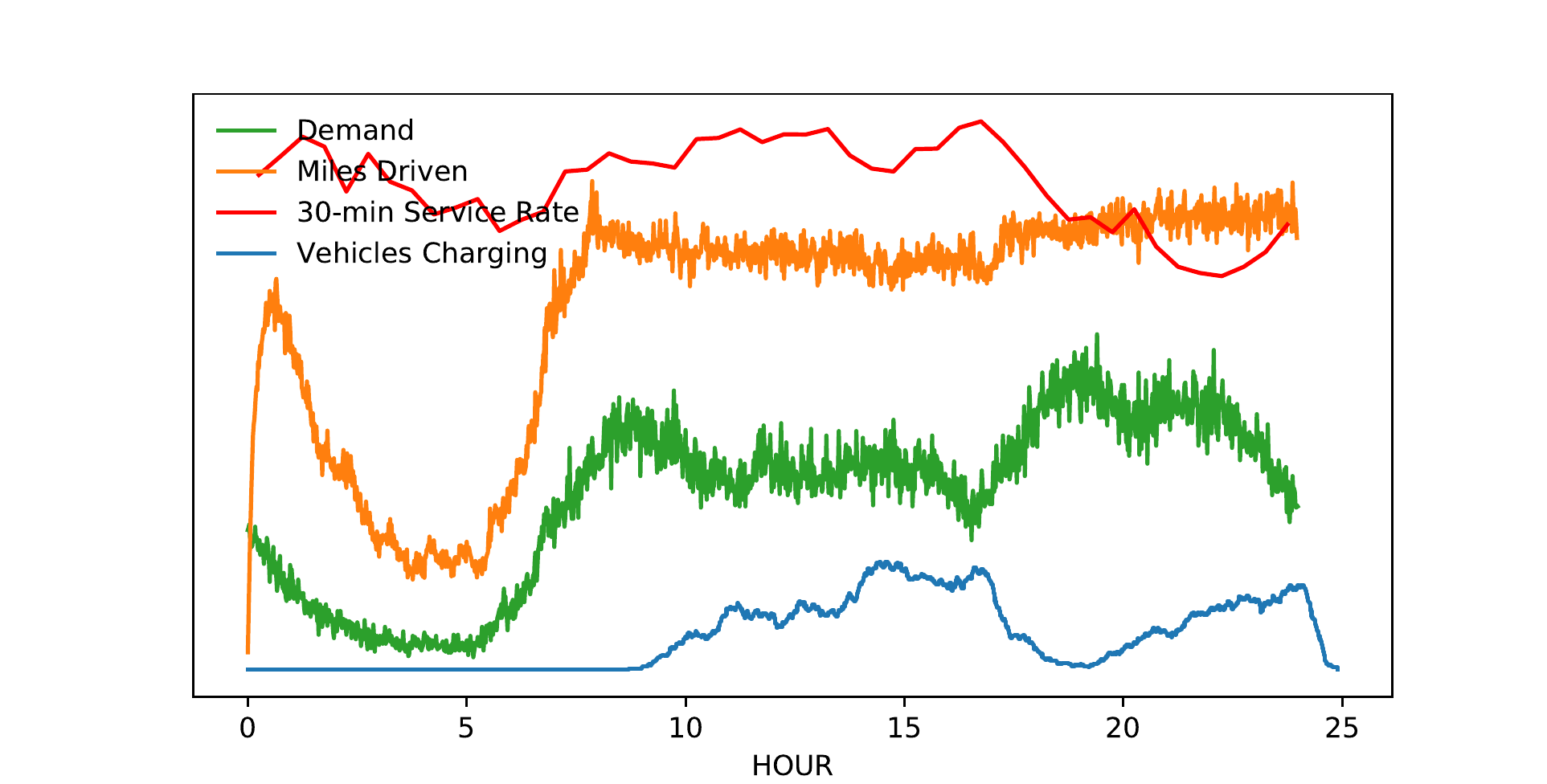}
  \caption{Benchmark method on 1000 vehicles starting at midnight.  By coincidence it turns out that the vehicles all finish charging before the day's peak demand, a fortunate gain for this particular run.}
  \label{figure:1000_naive_midnight}
\end{figure}

The computation times for the experiments are shown in Figure \ref{table:compute_time}.  As expected, the average computation time increased as the number of vehicles increased.  In the case of 220 vehicles, the per iteration cost using the MILP long horizon method took significantly more time since there were many iterations where the solver would spend the maximum allowed time performing branch and bound.  In general, the long horizon heuristic and benchmark method did not add significant computational overhead.  In fact, there is a counterintuitive effect where the benchmark method seems to be spending less time per iteration even compared to the ICE method.  The changes in time are primarily driven by changes in the length of time required to solve the passenger assignment problem, which is affected by the number of vehicles that are online and available at different times throughout the day.

\begin{table}
    \caption{Average Per Iteration Computation Time.  Computation time increases with number of vehicles.  MILP methods take a significant amount of time.  The heuristic, benchmark, and ICE methods use roughly the same time.  The changes in time are primarily driven by changes in the length of time required to solve the passenger assignment problem.  This is affected by the number of vehicles that are online and available at different times throughout the day and the relation of that to the number of requests that need to be assigned.}
    \label{table:compute_time}
    \centering
    \begin{filecontents*}{tables/average_computation_time.csv}
,220,1000,2000
MIP (180 sec),4.41 sec,,
MIP (55 sec),1.95 sec,,
Heuristic,0.38 sec,6.14 sec,33.3 sec
Benchmark,0.35 sec,5.31 sec,29.7 sec
ICE,0.37 sec,5.79 sec,31.9 sec
\end{filecontents*}
\csvautotabular{tables/average_computation_time.csv}
\end{table}

\subsection {Sensitivity Analysis}

The results presented above require many assumptions about the practical setting in which a high capacity electric ridepool fleet might operate.  To add credence to the conclusions suggested by the results, here we will perturb a variety of the input parameters and discuss how the outputs change compared with the default settings.  First we vary the preferred availability schedule $A(t)$,  then we look at battery capacity, estimated range, number and location of charge stations, maximum distance the benchmark method can send vehicles to get to a charge station, and vehicle capacity.

Since the proposed methods rely heavily on the choice of the availability function, first we compare three different choices of $A(t)$: the preferred availability schedule, a uniform availability schedule, and an intentionally counterproductive schedule that emphasizes charging during peak periods.  For the heuristic method, the service rate is not very sensitive to the choice of availability function.  As shown in Tables \ref{table:220_rt_sensitivity} and \ref{table:1000_rt_sensitivity}, the impact of various choices of availability function yield similar results.  Choices that place too much availability at the wrong time of day do moderately worse.  All experiments had sufficient availability for the system to maintain fully charged vehicles.  This suggests that the majority of benefits of the heuristic charging method over the benchmark, which only served around 82.6\% and 83.4\% of requests in the case of 220 and 1000 vehicles, respectively, come from its ability to look ahead and appropriately spread the timing of vehicle charging.  As a secondary effect, choosing the availability function with consideration of demand patterns offer smaller, positive improvements.

\begin{table}
\caption{220 Vehicle Simulations for Availability Sensitivity.  Even in the worst case with an inverted schedule, the service rate is significantly higher than the benchmark method, at 82.6\% (see Table \ref{table:220_general_results}).}
\label{table:220_rt_sensitivity}
\centering
\begin{filecontents*}{tables/4-220-availability.csv}
$R(t)$ version,Rate,WT,RT,Delay,Abs,Rider,Shared,Distance
Typical,89.35,126.88,541.47,94.16,1.34,1.67,93.21,248.73
Uniform Demands,88.47,125.99,537.28,90.6,1.32,1.66,92.89,249.12
Inverted Schedule,88.9,125.61,539.39,92.7,1.33,1.67,92.84,247.49
\end{filecontents*}
\csvautotabular{tables/4-220-availability.csv}
\end{table}

\begin{table}
\caption{1000 Vehicle Simulations for Availability Sensitivity.  While generating $R(t)$ using an inverted schedule yields the lowest service rate, it still only loses 0.6\% service rate compared to the typical model we developed in section \ref{subsec:availability}, much smaller than the 2.5-3\% gap between all of these results and the benchmark method.}
\label{table:1000_rt_sensitivity}
\centering
\begin{filecontents*}{tables/5-1000-availability.csv}
$R(t)$ version,Rate,WT,RT,Delay,Abs,Rider,Shared,Distance
Typical,86.38,177.71,796.34,196.81,2.06,2.43,98.17,327.51
Uniform Demands,86.05,177.54,796.41,198.04,2.05,2.44,98.16,324.23
Inverted Schedule,85.8,176.37,793.63,196.1,2.04,2.42,98.09,325.7
\end{filecontents*}
\csvautotabular{tables/5-1000-availability.csv}
\end{table}

Next we vary the battery range used in the simulations, whose default value was 180km.  In Table \ref{table:battery_distance}, simulation results are shown for three different battery capacities: 150km, 180km, and 210km.  In general, longer battery ranges resulted in better performance.  This is not surprising since we know that ICE vehicles, effectively vehicles with infinite charge, have the highest performance.  There are some unexpected trends in the data under the benchmark method.  While the general trend preferring longer battery range appears for the high and low end of the battery range, with the default 180km we do not get strictly better or worse performance than a longer or shorter battery capacity.  This is likely due to the effects of timing since the benchmark method is sensitive to the dictates of the timing of battery depletion.  If by coincidence the vehicles run out of charge during a period of lighter demand there is less impact than if they run out during high demand.  So in some sense, the battery capacity determines the time of day the vehicle runs out of power.  The heuristic methods seem less dependent on this and thus has fairly steady variations in performance when the range is perturbed.

\begin{table}
\caption{Vehicle Simulations Across Various Battery Ranges.  Longer battery range typically resulted in better service rate since vehicles spent less time on average offline.  The benchmark method did not strictly follow this trend, however, since it is also sensitive to the demand patterns at the time of day vehicles typically need to begin charging.}
\label{table:battery_distance}
\centering
\begin{adjustbox}{width = \textwidth}
\begin{filecontents*}{tables/6-battery-ranges.csv}
Veh. Count,Method,Battery Range,Rate,WT,RT,Delay,Abs,Rider,Shared,Distance
220,Heuristic,150km,88.83,126.63,541.79,95.55,1.33,1.68,93.49,245.85
220,Heuristic,180km,89.35,126.88,541.47,94.16,1.34,1.67,93.21,248.73
220,Heuristic,210km,89.67,126.5,538.8,91.5,1.34,1.66,93.28,250.7
1000,Heuristic,150km,85.44,177.32,796.77,199.2,2.04,2.45,98.07,321.92
1000,Heuristic,180km,86.38,177.71,796.34,196.81,2.06,2.43,98.17,327.51
1000,Heuristic,210km,86.76,177.42,797.72,197.27,2.08,2.43,98.17,329.63
220,Benchmark,150km,83.03,125.93,536.04,90.74,1.23,1.66,92.85,233.46
220,Benchmark,180km,82.56,125.3,532.56,88.39,1.22,1.64,91.83,233.22
220,Benchmark,210km,84.31,125.66,532.86,89.7,1.25,1.66,92.65,235.65
1000,Benchmark,150km,82.63,176.1,793.19,199.75,1.97,2.46,98.38,310.12
1000,Benchmark,180km,83.35,175.79,791.68,196.61,1.98,2.44,98.22,315.77
1000,Benchmark,210km,82.9,175.88,788.48,194.79,1.96,2.4,98.17,317.72
\end{filecontents*}
\csvautotabular{tables/6-battery-ranges.csv}
\end{adjustbox}
\end{table}

Next, recall the heuristic method is also dependent on the estimate for the expected duration that vehicle batteries last for.  In Table \ref{table:battery_estimates} we compare the default 13 hour estimate for 220 vehicle and default 10 hour estimate for 1000 vehicles used in the simulations, as discussed in Section \ref{estimatedbattery}, with a variation of 2 hours, up and down.  The sensitivity appeared fairly low.  Although service rate increases with longer expected battery duration, recall that this means that it is more likely that vehicles will approach true zero charge, which is below what we consider to be zero in the planning model, based on the assumption of a buffer (Assumption \ref{assumption:qmin}).

\begin{table}
\caption{Analysis of Heuristic Method at Various Estimated Battery Lives.  Estimated battery life has a weak effect on service rate.}
\label{table:battery_estimates}
\centering
\begin{filecontents*}{tables/8-battery-estimates.csv}
Veh. Count,Battery Life Estimate,Rate,WT,RT,Delay,Abs,Rider,Shared,Distance
220,11hr,89.18,126.3,541.28,94.45,1.34,1.67,93.4,248.6
220,15hr,89.56,125.93,540.68,93.59,1.34,1.66,93.53,250.35
1000,8hr,85.79,177.6,796.32,198.08,2.05,2.44,98.12,323.64
1000,12hr,86.6,177.43,796.91,197.07,2.07,2.43,98.13,328.27
\end{filecontents*}
\csvautotabular{tables/8-battery-estimates.csv}
\end{table}

We consider sensitivity of the system to the placement of charge station infrastructure in two ways.
First, to test the sensitivity of the default charge station location layout, we compare the scenarios with 220 vehicles and 1000 vehicles to a layout made by a simple greedy algorithm that places stations of capacity one.  First we randomly place the first charge station, then subsequently place stations so that the each placed station maximizes the distance to the nearest placed charge station.  The results are shown in Tables \ref{table:220_stations} and \ref{table:1000_stations}, which order the result in the same order as they were presented for the default charge station placement scenarios shown in Tables \ref{table:220_general_results} and \ref{table:1000_results}, respectively.  Generally, the performance of the heuristic and MILP methods we not affected by the change in charge station placement; in the case of 220 vehicles the service rate stayed roughly the same or increased marginally, in the case of 1000 vehicle the service rate decreased marginally.  The benchmark method had a slightly larger change benefiting by about 0.6\% in the case of 220 vehicles and losing about 0.5\% in the case of 1000 vehicles.  In all cases, the heuristic and MILP methods outperformed the benchmark method.

\begin{table}
\caption{Results for 220 vehicles when charge stations are placed using a greedy algorithm.  The benchmark method increased in performance by about 0.6\%, but the heuristic and MILP methods are largely unchanged.  See Table \ref{table:220_general_results} to compare these results with the default charge station layout.}
\label{table:220_stations}
\centering
\begin{filecontents*}{tables/11-220-altstations.csv}
Method,ILP Time Limit,Rate,WT,RT,Delay,Abs,Rider,Shared,Distance
ICE,,90.99,125.96,534.68,86.77,1.35,1.63,92.06,256.68
MILP,(180 seconds),89.36,126.45,540.86,93.36,1.34,1.67,93.24,247.78
MILP,(55 seconds),89.36,126.45,540.86,93.36,1.34,1.67,93.24,247.78
Heuristic,,89.46,126.36,540.85,93.49,1.34,1.67,93.31,248.3
Benchmark,,83.15,125.44,530.94,87.46,1.22,1.64,91.72,233.36
\end{filecontents*}
\csvautotabular{tables/11-220-altstations.csv}
\end{table}

\begin{table}
\caption{Results for 1000 vehicle when charge stations are placed using a greedy algorithm.  Both the heuristic and benchmark methods have lower service rate when compared to the default charge station layout, shown in Table \ref{table:1000_results}, with the benchmark method impacted slightly more.}
\label{table:1000_stations}
\centering
\begin{filecontents*}{tables/12-1000-altstations.csv}
Method,Rate,WT,RT,Delay,Abs,Rider,Shared,Distance
ICE,88.44,176.87,793.97,190.46,2.1,2.37,98.19,342.59
Heuristic,86.0,177.88,795.7,196.18,2.05,2.43,98.09,328.53
Benchmark,82.84,176.46,792.51,196.66,1.97,2.43,98.18,317.52
\end{filecontents*}
\csvautotabular{tables/12-1000-altstations.csv}
\end{table}

Second, we consider the sensitivity of the methods to the number of available charge stations.  In the above experiments there were 101 stations available for the 1000 vehicles.  In Table \ref{table:reduced_station_count} we compare the outputs from the experiments had that number been reduced to 90 charge stations, using the same procedure as outlined in section \ref{chargestationlocation} for placement.  The heuristic method is hardly affected by this change while the benchmark method looses more than a full percent of service rate under reduced stations.  This is expected since the greedy assignment algorithm will not be able to give all vehicles that need to charge immediately available assignments while the heuristic can look ahead and avoid congestion situations all together.

\begin{table}
\caption{1000 Vehicle Simulations with Reduced Station Count.  Reducing the number of charging stations has a larger impact on the service rate of the benchmark algorithm than the heuristic method.}
\label{table:reduced_station_count}
\centering
\begin{filecontents*}{tables/7-1000-reduced-stations.csv}
Method,Charge Station Count,Rate,WT,RT,Delay,Abs,Rider,Shared,Distance
Heuristic,Typical 101,86.38,177.71,796.34,196.81,2.06,2.43,98.17,327.51
Heuristic,90,86.23,177.38,796.79,197.19,2.06,2.43,98.15,327.52
Benchmark,Typical 101,83.35,175.79,791.68,196.61,1.98,2.44,98.22,315.77
Benchmark,90,82.12,175.87,790.67,196.65,1.94,2.44,98.22,311.64
\end{filecontents*}
\csvautotabular{tables/7-1000-reduced-stations.csv}
\end{table}

Next, we look at a parameter the benchmark method has for the maximum distance between a vehicle's position and the charge station it is assigned to.  This parameter controls the trade off between keeping vehicles near areas with high demand and making sufficient charging options available.  We use a default of 900 seconds of travel time, but we compare this with 600 second and 1200 second options.  The choice of 900 seconds outperformed both the 600 and 1200 second bound.
See Table (\ref{table:naive_distances}).

\begin{table}
\caption{Comparison of Benchmark Method Maximum Charge Station Travel Time.  Allowing the greedy station assignment algorithm to assign vehicles to stations nearer or farther away did not change the service rate enough to be competitive with the heuristic method.  See Tables \ref{table:220_general_results} and \ref{table:1000_results}.}
\label{table:naive_distances}
\centering
\begin{filecontents*}{tables/9-travel-time.csv}
Veh. Count,Greedy Search Bound,Rate,WT,RT,Delay,Abs,Rider,Shared,Distance
220,600 sec,80.9,125.48,532.08,89.13,1.19,1.65,92.13,225.03
220,1200 sec,82.48,125.42,532.15,87.97,1.22,1.64,91.79,233.82
1000,600 sec,83.28,176.09,792.16,196.9,1.98,2.44,98.25,314.69
1000,1200 sec,83.3,176.08,791.89,196.91,1.98,2.44,98.2,316.73
\end{filecontents*}
\csvautotabular{tables/9-travel-time.csv}
\end{table}

Finally, we consider the sensitivity of the system to our choice of using capacity 10 vehicles by comparing against scenarios where the capacity is reduced to 4.  We compared the vehicle capacity across all vehicles counts, shown in Tables \ref{table:cap4_220}, \ref{table:cap4_1000}, and \ref{table:cap4_2000}.  As expected, the service decreases when vehicle have lower capacity.  The impact of lower capacity was strongest among simulations with larger fleet sizes.  The general relationships between the ICE, MILP, heuristic, and benchmark methods appears to be the same in the capacity 4 case as it was in the default capacity 10 case.

\begin{table}
\caption{Results for 220 vehicles when capacity is reduced to 4.  Compare to Table \ref{table:220_general_results} which shows the same simulations when vehicle capacity is 10.  The ICE method suffers little impact from the change in capacity, but the other methods lose around a percentage point of service rate.}
\label{table:cap4_220}
\centering
\begin{filecontents*}{oldtables/1-220-vehicles.csv}
Method,ILP Time Limit,Rate,WT,RT,Delay,Abs,Rider,Shared,Distance
ICE,,90.32,125.96,535.6,213.93,1.07,1.76,91.76,262.55
MILP,(180 Seconds),88.15,126.67,541.04,221.03,1.05,1.8,92.72,253.9
MILP,(55 Seconds),88.22,126.89,542.11,222.36,1.05,1.81,93.08,253.85
Heuristic,,88.65,126.23,541.01,220.95,1.06,1.74,92.88,254.69
Benchmark,,81.72,125.43,533.75,214.76,0.96,1.79,91.42,239.56
\end{filecontents*}
\csvautotabular{oldtables/1-220-vehicles.csv}
\end{table}

\begin{table}
\caption{Results for 1000 vehicles when capacity is reduced to 4.  Compare to Table \ref{table:1000_results} which shows the same simulations when vehicle capacity is 10.  Roughly all methods in this case were negatively impacted by the lower capacity, typically losing about two and a half percentage points of service rate.}
\label{table:cap4_1000}
\centering
\begin{filecontents*}{oldtables/2-1000-vehicles.csv}
Method,Rate,WT,RT,Delay,Abs,Rider,Shared,Distance
ICE,86.17,183.73,784.4,369.87,1.62,1.83,97.57,345.8
Heuristic,83.99,183.99,786.45,376.5,1.58,1.99,97.61,329.38
Benchmark,80.83,183.21,780.52,374.55,1.51,2.4,97.58,318.73
\end{filecontents*}
\csvautotabular{oldtables/2-1000-vehicles.csv}
\end{table}

\begin{table}
\caption{Results for 2000 vehicles when capacity is reduced to 4.  Compare to Table \ref{table:2000_results} which shows the same simulations when vehicle capacity is 10.  All methods were negatively impacted by the lower capacity, with the ICE method taking the small hit to service rate with a 3.4\% loss, and the benchmark method taking the worst hit with a 4.5\% loss in service rate.}
\label{table:cap4_2000}
\centering
\begin{filecontents*}{oldtables/EXTRA-10-2000-vehicles.csv}
Method,Rate,WT,RT,Delay,Abs,Rider,Shared,Distance
ICE,89.73,180.95,770.48,349.7,1.66,2.07,97.71,338.21
Heuristic,88.74,183.06,775.29,358.8,1.65,2.56,97.76,328
Benchmark,85.57,182.14,768.51,358.11,1.58,2.48,97.76,315.49
\end{filecontents*}
\csvautotabular{oldtables/EXTRA-10-2000-vehicles.csv}
\end{table}

In summary, the sensitivity analysis has shown that changes in the input parameters can make marginal changes to the performance of the system in terms of service rate.  However, despite the differences in performances under the various perturbations a few features remained constant.  First, the ICE method always served as an upper bound for the performance of any system that required a charging scheme.  Second, the benchmark method was always outperformed by the MILP and heuristic methods, supporting the intuition that algorithms that foresee future conditions can be beneficial in electric ridepool systems.

\section {Conclusions}

Results from experiments demonstrated that high-capacity ridepool systems can be run at scale with electric vehicles in real time.  Unlike the case with ICE vehicles which do not need refueling over the course of the day, the requirement for electric vehicles to recharge added a new set of constraints to the already hard problem of high-capacity ridepooling.  This work presented a two-phase method that allowed the operator to manage the fleet by separating the problem of passenger assignment and charge scheduling into separate optimization procedures.  While our method required estimating new parameters, such as estimates of daily vehicle availability requirements, reasonable estimates we ran experiments on performed similarly and well suggesting such an approach may work well in a real-world setting.

The benefits of the methods presented are significant since strategies that act purely online without making any estimates of future demand, such as the benchmark method studied here, pose risk when the distribution of charges is initially uniform across vehicles.  These risks are most prominent when the entire fleet is fully charged at the beginning of the work day, though experimental results showed that night service can sometimes mitigate this risk by causing the charge in vehicles to spread.  The spread caused by night time service is different from regular service as the variability is not caused by the vehicles traveling on roads with different speeds but rather by the large portion of downtime experienced by the vehicles.  In cases when night traffic is especially low, as seen in the experiments with 220 vehicles, this effect is smaller.  In realistic settings, operators would likely use dynamic fleet sizing which would concentrate utilization of vehicles and again raise similar risks as with the fully charged fleet.

Although our experiments assumed the distribution of vehicle miles traveled is uniform over the course of the day, it may be beneficial not to make this assumption.  For example, if one were to operate a system 24 hours a day, the night time may be better accounted for with lower expected battery usage.  This change can easily be implemented in our framework since expected battery usage is entirely accounted for per-period; one could estimate and utilize separate $q_{est, t}$ in the long horizon MILP for each time period.

If the operator rents rather than owns charge station capacity, the model can further be adapted to allow for different numbers of vehicles to be assigned at different times of day.  This implementation has the caveat that the long horizon problem will only take into account whether the schedule is feasible for the total number of vehicles charging, not whether the schedule is feasible for the particular set of stations that are available at different times.  For example, suppose that before noon the operator has access to one set of stations and then after noon they have access to a mutually exclusive set of stations.  The solution to the long horizon problem may still schedule a vehicle to charge from 11:45 until 12:15 since sufficient capacity seems to exist, but it is not feasible since no station assignment could accomplish that.  While this problem can be avoided in both the heuristic and MILP formulations by forbidding scheduling vehicle times that cross certain time boundaries, this solution is too strong for most cases since it is rare that the time constraint would apply to all charging stations at a single time step.

In the future we would like to extend the model to allow for dynamic fleet sizing, where vehicles that do not need to be available are allowed to stop.  The model used in this paper assumes that vehicles must always be able to be in service unless they are charging.  But in practice it may be advantageous to have the capability to keep more vehicles online during busy hours while having fewer vehicles available during off peak hours.

\section{Acknowledgements}

The authors would like to thank Vindula Jayawardana for his help in coding the original simulator that was extended to produce the computational results in this paper.  This work was partially supported by the National Science Foundation (grant number 1839346) and the United States Department of Energy (grant number DE-EE0008464 subaward RQ19-119R06).

\bibliographystyle{acm} 
\bibliography{refs.bib}

\end{document}